\pdfoutput=1
\documentclass[graybox, vecphys]{svmult}

\usepackage{mathtools}
\usepackage{tikz}
\usetikzlibrary{backgrounds}
\definecolor{lightblue}{rgb}{0.3,0.7,0.95}
\newcommand{\tree}[1]{\ensuremath \vec #1}

\usepackage{type1cm}        % activate if the above 3 fonts are
\usepackage{makeidx}         % allows index generation
\usepackage{graphicx}        % standard LaTeX graphics tool
\usepackage{multicol}        % used for the two-column index
\usepackage[bottom]{footmisc}% places footnotes at page bottom

\usepackage{newtxtext}       % 
\usepackage{newtxmath}       % selects Times Roman as basic font

\usepackage{algorithmic}
\usepackage[ruled,vlined, linesnumbered]{algorithm2e}
\newcommand{\R}{{\mathbb R}}
\newcommand{\C}{{\mathbb C}}

\newcommand{\N}{{\mathbb N}}
\newcommand{\T}{{\mathbb T}}

\DeclareMathOperator{\supp}{supp}     % support

\DeclareMathOperator*{\argmin}{argmin}

\newcommand{\K}{\mathbb{K}}
\makeindex             % used for the subject index
\usepackage{bibentry}
\nobibliography* 
\usepackage[resetlabels]{multibib}
\newcites{WE}{References}     %% replace by your initials e.g., LG -> HK
\begin{document}
	%!TEX root = HierarchicalCompressedSensing.tex
\title*{Hierarchical compressed sensing}
\titlerunning{Hierarchical compressed sensing}
\author{J. Eisert, A. Flinth, B. Gro\ss, I. Roth, G. Wunder}
\institute{Jens Eisert \at Dahlem Center for Complex Quantum Systems and
Department of Mathematics and Computer Science, Freie Universit\"at Berlin, Berlin, Germany.
\email{jense@physics.fu-berlin.de}
\and 
Axel Flinth \at Institute for Electrical Engineering, Chalmers University of Technology,
Gothenburg, Sweden. \email{flinth@chalmers.se}
\and
Benedikt Gro\ss \at Department of Mathematics and Computer Science, Freie Universit\"at Berlin, Berlin, Germany,
\email{benedikt.gross@fu-berlin.de} 
\and 
Ingo Roth \at 
Quantum Research Centre, Technology Innovation Institute, Abu Dhabi, UAE. \\ 
Dahlem Center for Complex Quantum Systems, 
Freie Universit\"at Berlin, Berlin, Germany. \\
\email{i.roth@fu-berlin.de}
\and
Gerhard Wunder \at Department of Mathematics and Computer Science, Freie Universit\"at Berlin, Berlin, Germany. \email{g.wunder@fu-berlin.de}}

\maketitle

\vspace{-3\baselineskip}

\abstract{Compressed sensing is a paradigm within signal processing that provides the means for recovering structured signals from linear measurements in a highly efficient manner. Originally devised for the recovery of sparse signals, it has become clear that a similar methodology would also carry over to a wealth of other classes of structured signals. In this work, we provide an overview over the theory of compressed sensing for a particularly rich family of such signals, namely those of hierarchically structured signals. Examples of such signals are constituted by blocked vectors, with only few non-vanishing sparse blocks. We present recovery algorithms based on efficient hierarchical hard-thresholding. The algorithms are guaranteed to converge, in a stable fashion both with respect to measurement noise as well as to model mismatches, to the correct solution provided the measurement map acts isometrically restricted to the signal class. We then provide a series of results establishing the required condition for large classes of measurement ensembles. Building upon this machinery, we sketch practical applications of this framework in machine-type communications and quantum tomography. }

\section{Introduction}
\let\thefootnote\relax\footnotetext{%
This book chapter is a report on some of the findings of the DFG-funded project
 EI 519/9-1 within  the priority 
 program `Compressed Sensing in Information Processing' (CoSIP).}
The field of compressed sensing studies the recovery of structured signals from linear measurements
\citeWE{FouRau13,CompressedSensingGitta}.  
Originally focusing on the structure of sparsity 
of vectors, the framework was quickly extended to the structure of low-rankness of matrices. These structures are simultaneously restrictive and rich. 
They are restrictive so that they allow for signal recovery using far fewer linear measurements than the ambient dimensions suggests and rich in that they naturally appear in a plethora of applications.
That being said,
 in many practically relevant applications, the signals feature a more restrictive structure than mere sparsity or low-rankness. 
A particularly important broad class arising in a wealth of contexts are \emph{hierarchically structured signals}. Such structures are in the focus of this book chapter. 

The perhaps simplest examples are constituted by \emph{hierarchically sparse vectors}.  
A two-level hierarchically sparse vector is a vector consisting of multiple blocks with a restricted support as follows: 
Only a small number of the blocks have non-vanishing entries and the blocks are themselves sparse. An illustrative example can be given via imagining a telecommunication base station responsible for handling a large set of potential users. If in each instance, only a few users actively transmit, and the messages that are transmitted are sparsely representable, the vector compiling all messages in its blocks is hierarchically sparse.
The hierarchically sparse vectors will serve as the main illustrative example of the entire chapter. 
It is straight-forward to generalize this notion for vectors with a hierarchy of nested blocks with sparsity assumptions restricting the number of non-vanishing blocks on each level. 

Another hierarchical structure of interest is given by replacing the sparsity constraint on the vector-valued blocks by a low-rank assumption of matrix-valued blocks. A motivating example here can be found in quantum tomography, where quantum states can be modelled as low-rank Hermitian matrices. 
Hierarchical structures of quantum states arise here in the tasks  
of performing quantum tomography with a partially uncalibrated measurement device or de-mixing sparse sums of quantum states.

An intriguing feature of hierarchically structured signals is that their recovery task is amenable to efficient thresholding algorithms. 
In general, thresholding algorithms such as the iterative hard-thresholding pursuit are built on the insight that, in contrast to the original recovery problem, the projection onto the set of structured signals is efficient and in fact often particularly simple. 
This allows one to employ algorithmic strategies such as projective gradient descent. 

For hierarchically sparse signals, it turns out that the calculation of the projection has the same computational complexity as the thresholding onto sparse signals. 
Furthermore, the hierarchical structure allows for the parallelization of the projections for the blocks on each level yielding potential for further reducing the time complexity by exploiting the restrictive structure. 
Based on this insight, we formally introduce variants of the \emph{iterative hard-thresholding (IHT) algorithm} and the \emph{hard-thresholding pursuit (HTP)} 
for hierarchically sparse signals.

For the IHT and HTP algorithm, recovery guarantees for measurement maps that act close to isometrically, on sparse vectors, exist.  
Due to their similarity, the recovery algorithms for hierarchically sparse signals inherit the recovery guarantees from the original IHT and HTP provided the measurement map exhibits a \emph{restricted isometry property (RIP)} that is adapted to the hierarchically structured signal set.  We refer to the modified RIP restricted to hierarchically sparse signals as the \emph{hierarchically restricted isometry property (HiRIP)}.

In this chapter, we derive a series of theoretical results concerning the HiRIP. Requiring only HiRIP instead of RIP for the measurement opens up the possibility of 
exploiting multiple benefits. 
First, standard measurement ensembles such as random Gaussian matrices can achieve HiRIP with a reduced sampling complexity compared to RIP.  
The achievable logarithmic improvement mirrors the reduced complexity of the restricted signal set compared to standard sparse vectors. 
Second, we introduce an ensemble of operators that do have the HiRIP, but not RIP in any parameter regime. 
We give this flexible class of operators the name \emph{hierarchical measurements}, since they are naturally adapt to hierarchical structures. 
Hierarchical measurements combine different measurement maps on each level of the hierarchy and, as we show, inherit HiRIP from standard RIP and coherence properties of their constituent maps. 
An important instance of hierarchical measurement are Kronecker-products of measurements such that each factor acts on the blocks of a certain hierarchy level. 

Finally, we illustrate how the framework of hierarchical compressed sensing can be applied in applications in machine-type communications and quantum technologies providing motivating examples and evaluations of practical performances.

Let us end with an outline of the remainder of the chapter. In   Sec.~\ref{section:signal_model} and Sec.~\ref{sec:hiRecoveryAlgs}, respectively, we formally introduce hierarchically sparse vectors, and present the algorithms used to recover them. 
Sec.~\ref{sec:hiRIP} is devoted to theoretical results concerning the hierarchical restricted isometry property (HiRIP) and step-by-step develops a flexible toolkit to establish the HiRIP for large classes of measurement ensembles. 
In Sec.~\ref{sec:sparseLowRank}, we move on to discussing the sparse, low-rank signal model, including how the recovery algorithms can be adapted. In Sec.~\ref{section:applications}, we provide a more specific discussions of selected applications. 
We close with a conclusion as well as a small outlook in  Sec.~\ref{sec:outlook}.

\section{Hierarchically sparse vectors}
\label{section:signal_model}

We consider structured sparse signals that are vectors over the field $\K$, referring to either the reals $\R$ or the complex numbers $\C$, and are hierarchically structured into blocks. The support is restricted by sparsity assumptions on one or multiple hierarchy levels. The simplest instance of hierarchically sparse signals are two-level hierarchically sparse vector with constant block-sizes and sparsities \citeWE{SprechmannEtAl:2010,FriedmanEtAl:2010, SprechmannEtAl:2011, SimonEtAl:2013}. 

\begin{definition}[Two-level hierarchically sparse vectors]
Let $N,n,s,\sigma\in\N$. 
A vector $ x\in\K^{Nn}$ is called \emph{$(s,\sigma)$-hierarchically sparse}, 
if it consists of blocks $ x_i\in\K^n$, $ x^\top= ( x^\top_1,\ldots,  x^\top_i,\ldots, x^\top_N)^\top$, 
where at most $s$ blocks $ x_i$ are non-zero, 
and each of the non-zero blocks are at most $\sigma$-sparse.
\end{definition}

For brevity, we write \emph{$(s,\sigma)$-sparse}, dropping the \emph{hierarchically} in the following. 
We refer to the set of $(s,\sigma)$-sparse vectors in $\K^{Nn}$ as $\mathcal{S}^{N,n}_{s,\sigma}$ or simply $\mathcal{S}$ if the parameters are clear from the context. 
In Fig.~\ref{WEfig:HiThresholding} (e), an illustration of a $(2,2)$-sparse vector with $N=5$ blocks and block-size $n=7$ is depicted. 
We also call the support $\supp( x) \subset [N] \times [n]$ of an $(s,\sigma)$-sparse vector a \emph{$(s,\sigma)$-sparse support}, where $[n] \coloneqq \{1, \ldots, n\}$. 
The definition of a two-level hierarchically sparse vectors can be generalized in several directions: 
We can allow different block sizes and block sparsities.
Furthermore, each block is allowed to be a hierarchically sparse vector itself. 
This gives rise to a more general recursive definition of hierarchically sparse vectors with arbitrary many levels. 
The defining data of such a general hierarchically sparse vector can be collected in a rooted tree consisting of nodes, labeled by block-sizes and -sparsities, see Fig.~\ref{WEfig:HiSparse}. 
We refer to Ref.~\citeWE{RothEtAl:2020:HiHTP} for a formal definition of general hierarchically sparse vectors. 
Other special cases of hierarchically sparse vectors have been considered in the literature. Prominent examples are \emph{block sparse} \citeWE{EldarMishali:2009a,EldarMishali:2009b} and \emph{level sparse} \citeWE{AdcockEtal:2013,LiAdcock:2016} vectors.

\begin{figure}[tb]
 	\centering
 	\input{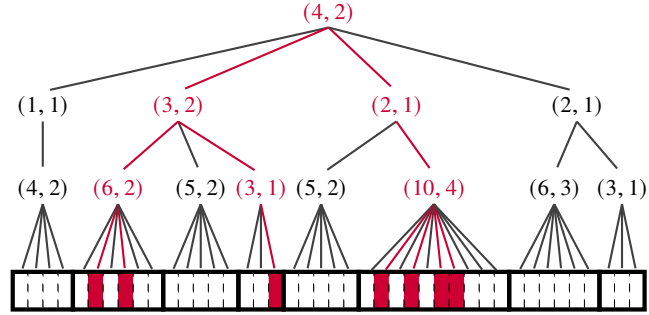}
 	\caption{This figure shows an example of a hierarchically sparse vector. The grouping of the entries is encoded in a rooted tree. The children of a vertex constitute a block at their level. The pair of values at each vertex indicates the block-size (number of children) and the sparsity, i.e.\ the number of children with non-vanishing entries. The leaves of the tree are identified with the entries of the vector. The support of the vector drawn below and corresponding vertices with non-vanishing entries are highlighted in red.  \textcopyright 2020 IEEE. Reprinted, with permission, from  Ref.~\protect\citeWE{RothEtAl:2020:HiHTP}.}
	\label{WEfig:HiSparse}
\end{figure}

Another setting where the hierarchical sparsity naturally emerges is so-called \emph{bi-sparsity}, see, e.g.,
Ref.\ \citeWE{foucart2019jointly}. 
In said reference, a Hermitian matrix $X\in \K^{n \times n}$ is called bisparse if there exists a set $S \subseteq [N]$ with $\vert{S}\vert\leq s$ so that $X_{ij}$ is non-zero, only if both $i$ and $j$ are in $S$. Clearly, any bisparse matrix can be interpreted as an $(s,s)$-sparse vector. 
More generally, a matrix $Y \in \K^{N \times n}$ with $Y_{ij}$ non-zero for $i\in S$ and $j \in \Sigma$ for sets with cardinalities $\vert S \vert =s,\vert$, $\Sigma \vert\leq \sigma$ can be regarded as $(s,\sigma)$-bisparse, and in the same manner identified with an $(s,\sigma)$-sparse vector. 
Bi-sparsity is of course more restrictive than hierarchical sparsity, but the projection operator onto the set of bisparse matrices is -- in stark contrast to its hierarchical sparsity counterpart -- NP-hard to compute. Hierarchical sparsity can thus be seen as a relaxation of bi-sparsity which allows for more efficient recovery procedures. We refer to Ref. \citeWE{foucart2019jointly} for a more comprehensive discussion on these matters, as well as other ways to relax the bisparse structures. 
We encounter this relaxation in conjunction with blind deconvolution in Sec.~\ref{sec:blinddeconv}, and a non-commutative analog of it in our discussion of blind quantum tomography in Sec.~\ref{sec:blind_quantum_state_tomography}.

For simplicity and notational clarity, we content ourselves to present the framework for two-level hierarchically sparse vectors. %
It is straight-forward to generalize the algorithmic strategies and most analytical results of this chapter to the general definition of hierarchically sparse vectors outlined above, see Ref.~\citeWE{RothEtAl:2020:HiHTP} for details. 

\section{Hierarchical thresholding and recovery algorithms}
\label{sec:hiRecoveryAlgs}

We study the linear inverse problem of recovering an $(s, \sigma)$-hierarchically sparse vector $ x \in \K^{Nn}$ from noisy linear measurements of the form
\begin{equation*}
    \label{eq:measurement}
     y =  M  x  +  e,
\end{equation*}
where $ M \in \K^{m \times Nn}$ is the linear measurement operator and $\ e \in \K^{m}$ encodes additive noise.  
The recovery task can be cast as the constraint optimization problem
\begin{equation}
    \label{eq:hi_sparse_problem}
    \operatorname*{minimize}_{ x \in \K^{Nn}}\ \frac{1}{2}\| y-  M x\|^2 \quad \text{ subject to  $x$ is $(s,\sigma)$-sparse.},
\end{equation}
where $\| y\| = \left[\sum_{i} |y_i|^2 \right]^{1/2}$ denotes the $\ell_2$-norm. 

\begin{figure}[tb]
	\centering
	\input{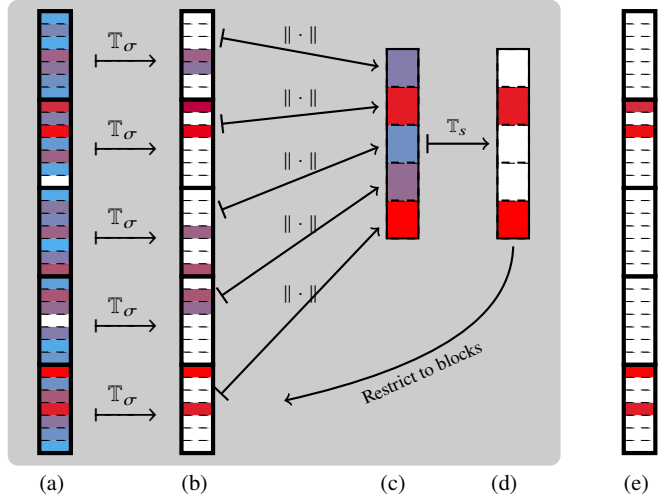}
	\caption{In this figure, the evaluation of the hierarchical thresholding operator $\T_{s,\sigma}$ is illustrated. Starting with a given dense vector (a), each block is thresholded to its best $\sigma$-sparse approximation (b). To determine the $s$ dominant blocks, the $\ell_2$-norm is calculated for each block. The resulting vector (c) of length $N$ is again thresholded to its best $s$-sparse approximation (d). The resulting blocks indicated by the $s$-sparse approximation (d) are selected from the $\sigma$-sparse approximation (b). The remaining $(s,\sigma)$-sparse support (e) is the output of $\T_{s,\sigma}$. 
	\textcopyright 2020 IEEE. Reprinted, with permission, from  Ref.~\protect\citeWE{RothEtAl:2020:HiHTP}.}
	\label{WEfig:HiThresholding}
\end{figure}

So called hard-thresholding algorithms solve the analogous optimization problem to \eqref{eq:hi_sparse_problem} for standard $s$-sparse recovery by making use of the projection of a vector $ z\in\K^n$ onto the set of $s$-sparse vectors.  
The projection onto $s$-sparse vectors, 
\begin{equation*}
    \label{eq:hard_thresholding}
    \T_s({z}) := \argmin\limits_{{x}\in\K^n} \| x- z\| \quad \text{ subject to $x$ $s$-sparse}, 
\end{equation*}
can be computed efficiently via hard thresholding, i.e.\ by setting to zero all but the $s$ largest entries in absolute value.
Note that since the set of $s$-sparse vectors in $\K^n$ is not a convex set, the projection is non-unique.  
But for the arguments made here every solution works equally well. 
Using a quick-select algorithm \citeWE{Hoare:1961}, the average computational complexity of the thresholding operation is in $O(n)$ with worst-case complexity $O(n^2)$.  

Following the blue-print of model-based compressed sensing \citeWE{BarCevDua10}, we can derive variants of standard hard-thresholding algorithms for the more restrictive sparsity structure under consideration here by modifying the thresholding operator accordingly. 
The projection of a vector ${z}\in\K^{Nn}$ onto the set $\mathcal{S}$ of $(s, \sigma)$-hierarchically sparse vectors, 
\begin{equation*}
    \label{eq:projection_operator}
    \T_{s,\sigma}({z}) = \argmin\limits_{{x}\in\mathcal{S}} \frac12 \| x- z\|^2\ ,
\end{equation*}
can be computed via \emph{hierarchical hard thresholding}: 
First, the standard hard thresholding operation $\T_\sigma$ is applied to each block.
Then, all but the $s$ blocks with largest $\ell_2$-norm are set to zero. 
The procedure is summarized as Alg.~\ref{WEalg:hierarchical_hard_thresholding} and illustrated in Fig.~\ref{WEfig:HiThresholding}. 
We find that the average computational complexity of the hierarchical thresholding operation scales as $O(Nn)$, i.e.\ linear in the overall vector space dimension as for the standard hard thresholding. 
Furthermore, the  hard thresholding and $\ell_2$-norm calculation of the different blocks can be parallelized, reducing the time-complexity to $O(\max(N, n))$. 
The hierarchical thresholding operation can be extended  recursively to general hierarchically sparse signals described in Sec.~\ref{section:signal_model} without increasing the overall computational complexity.

\begin{algorithm}[tb]
\label{WEalg:hierarchical_hard_thresholding}
\SetAlgoLined
\DontPrintSemicolon
\SetKwInOut{Input}{input}\SetKwInOut{Output}{output}
\SetKwInOut{Init}{initialize}\SetKwFunction{Break}{break}
\Input{$ z\in\K^{Nn}$, sparsity levels $(s, \sigma)$ \;} 
\For{$i \in [N]$}{
${x}_i = \T_\sigma({z}_i)$; \;
$v_i  = \| x_i\|$; \;
}
$I = \supp\left(\T_s\left((v_1, \ldots, v_N)\right)\right)$; \;
\For{$i \in [N] \setminus I$}{
${x}_i = 0$ } 
\Output{$(s, \sigma)$-hierarchically sparse vector ${x} = (x^\top_1, \ldots,  x_N^\top)^\top$}
\caption{Hierarchical hard thresholding}
\end{algorithm}

Equipped with an efficient thresholding operation, we can 
formulate recovery algorithms for hierarchically sparse signals following standard strategies. 
A particularly simple algorithm is the \emph{iterative hard thresholding algorithm} \citeWE{BlumensathDavies:2008} which performs a projected gradient descent.  
The resulting \emph{hierarchical iterative hard-thresholding algorithm} (HiIHT, Alg.~\ref{WEalg:hiiht} \citeWE{wunder2019low}) alternates gradient descent steps of the objective function \eqref{eq:hi_sparse_problem} with the hard-thresholding operation $\T_{s,\sigma}$.

\begin{algorithm}[tb]
\label{WEalg:hiiht}
\SetAlgoLined
\DontPrintSemicolon
\SetKwInOut{Input}{input}\SetKwInOut{Output}{output}
\SetKwInOut{Init}{initialize}\SetKwFunction{Break}{break}
\Input{data ${y}\in\K^m$, measurement operator ${M}\in\K^{m\times Nn}$, sparsity levels $(s,\sigma)$}
\Init{${x}^{(0)} = 0$}
\Repeat{stopping criterion is met at $t = t^\ast$}{
$\bar{{x}}^{(t)} = {x}^{(t-1)} + \tau^{(t)} {M}^\ast\left(y-{M}{x}^{(t-1)}\right)$; \;
${x}^{(t)} = \T_{s,\sigma}\left(\bar{{x}}^{(t)} \right)$; \;
}
\Output{$(s,\sigma)$-sparse vector $ x^{(t^\ast)}$}
\caption{HiIHT algorithm}
\end{algorithm}

Here, $\tau^{(t)}$ is a suitably chosen stepsize. 
The original IHT algorithms uses constant steps $\tau^{(t)} = 1$ for all $t$. 
Alternative strategies include backtracking as in the normalized iterative hard thresholding (NIHT) algorithm \citeWE{BlumensathDavies:2009}.

Faster convergence can be achieved with an adaption of the hard-thresholding pursuit \citeWE{Foucart:2011} to hierarchical sparsity, the HiHTP \citeWE{RothEtAl:2016:Proceedings, RothEtAl:2020:HiHTP}.  
Compared to the HiIHT, the HiHTP algorithm uses the result of the thresholded gradient-step as a proxy to guess the support of the correct solution in each step. 
Subsequently, a linear least-squares problem is 
solved on the support guess. 
The solution can be computed via pseudo-inverse or an approximate method. 
Notably, with this modification, if the algorithm finds the correct solution, it does this in a finite number of steps to the precision of the least-squares problem solver.
The HiHTP algorithm is given as Alg.~\ref{WEalg:hihtp}.

\begin{algorithm}[tb]
\label{WEalg:hihtp}
\SetAlgoLined
\SetKwInOut{Input}{input}\SetKwInOut{Output}{output}
\SetKwInOut{Init}{initialize}\SetKwFunction{Break}{break}
\DontPrintSemicolon
\Input{data ${y}\in\K^m$, measurement operator ${M}\in\K^{m\times Nn}$, sparsity levels $(s,\sigma)$}
\Init{${x}^{(0)} = 0$}
\Repeat{stopping criterion is met at $t = t^\ast$}{
$\bar{x}^{(t)} = x^{(t-1)} + \tau^{(t)} M^*\left(y-{M}{x}^{(t-1)}\right)$; \;
$I^{(t)} = \supp\left(\T_{s,\sigma}\left(\bar{{x}}^{(t)} \right)\right)$; \;
${x}^{(t)} = \argmin\limits_{{x}} \frac12 \|{y}-{M}{x}\|^2 \quad$ subject to $\quad \supp(x)\subseteq I^{(t)}$; 
}
\Output{$(s,\sigma)$-sparse vector $ x^{(t^\ast)}$}
\caption{HiHTP algorithm}
\end{algorithm}

The computational complexity of both algorithms, HiIHT and HiHTP, is typically dominated by the matrix vector multiplication with the measurement matrix $ M$ and $ M^\ast$, scaling in general as $O(mNn)$. If a fast matrix vector multiplication is available for the measurement matrix, this scaling can be significantly improved.  

The additional least-square solution in the HiHTP algorithm contributes $O(s\sigma m^2)$ operations.
For this reason, HiIHT can be faster per iteration than the HiHTP in certain parameter regimes.  
Note that the computational complexity, featuring the overall vector space dimension $Nn$ and the total sparsity $s\sigma$, is identical to the complexity of the original IHT and HTP algorithms. 

Modifications using hierarchically sparse thresholding can also be applied to other compressed sensing algorithms such as the \emph{CoSAMP}~\citeWE{Needell08}, the \emph{Subspace Pursuit}~\citeWE{DaiMilenkovic:2009} or 
\emph{Orthogonal Matching Pursuit}, see e.g., Refs.\  \citeWE{Tropp:2004,LiuSun:2011} and references therein.
Proximal operators of the convex relaxations of the problem \eqref{eq:hi_sparse_problem} can be calculated using soft-thresholding operations yielding a hierarchical version of the LASSO algorithms \citeWE{SprechmannEtAl:2010,FriedmanEtAl:2010, SprechmannEtAl:2011, SimonEtAl:2013}. 
Due to their similarity, the HiIHT and HiHTP algorithms inherit their convergence proofs and recovery guarantees  with slight modifications from their non-hierarchical counterparts. 
To this end, we make use of the variant of the \emph{restricted isometry property (RIP)} \citeWE{candes2005decoding} adapted to hierarchically sparse signals. 

\begin{definition}[Hierarchical restricted isometry property (HiRIP)]
Given a linear operator ${M}: \K^{Nn}\rightarrow \K^m$, we denote by $\delta_{s,\sigma}$ the smallest constant such that 
\begin{equation*}
    (1-\delta_{s,\sigma})\|{x}\|^2 \leq \|{M}{x}\|^2\leq (1+\delta_{s,\sigma})\|{x}\|^2
\end{equation*}
holds for all $(s,\sigma)$-hierarchically sparse vectors ${x}\in\K^{Nn}$.
\end{definition}

We will also refer to the standard $s$-sparse RIP constant $\delta_s$, defined analogously with the bounds holding for all $s$-sparse vectors. 
The standard RIP constant dominates the HiRIP constant as $\delta_{s\sigma} \geq \delta_{s, \sigma}$ since $\mathcal{S}_{s, \sigma}$ is a subset of the set of $s\cdot \sigma$-sparse vectors. 
But as we will see below, using the HiRIP allows for a considerably more fine-grained analysis, yielding improvements in the sampling complexity.

In terms of a HiRIP condition, we can guarantee a robust and stable convergence to the correct solution for the hierarchical hard thresholding algorithms. 
To this end, given $ x \in \K^{Nn}$ and a support set $\Omega \subset [N] \times [n]$, we denote by ${x}\rfloor_\Omega$ the projection of $ x$ onto the subspace of $\K^{Nn}$ indicated by $\Omega$. 

\begin{theorem}[Recovery guarantee for HiIHT and HiHTP \citeWE{wunder2019low,RothEtAl:2020:HiHTP}]
Suppose the measurement operator ${M}: \K^{Nn}\rightarrow \K^m$ satisfies the HiRIP condition
\begin{equation*}
    \delta_{3s, 2\sigma}< \delta_*,
\end{equation*}
where $\delta_*$ is a threshold, equal to ${1}/{\sqrt{3}}$ for the HiHTP-algorithm and equal to $\sqrt{2}-1$ for the HiIHT-algorithm. Then, for ${x}\in\K^{Nn}$, ${e}\in\K^m$ and $\Omega\subset[N]\times [n]$ an $(s,\sigma)$-hierarchically sparse support set, the sequence $({x}^k)_k$ defined by by HiIHT (Alg.~\ref{WEalg:hiiht}) or HiHTP (Alg.~\ref{WEalg:hihtp}), respectively, with ${y} = {M}{x}\rfloor_\Omega+{e}$ satisfies, for any $k\geq 0$,
\begin{equation*}
    \|{x}^k-{x}\rfloor_\Omega\|\leq \rho^k\|{x}^0-{x}\rfloor_\Omega\| + \tau \|e\|,
\end{equation*}
where the constants $\rho$ and $\tau$ depend on which algorithm is used: For HiIHT
\begin{equation*}
    \rho^{\text{HiIHT}} =  \sqrt{3} \delta_{3s,2\sigma}, \quad \tau^{\text{HiIHT}} = \frac{2.18}{1-\rho^{\text{HiIHT}}},
\end{equation*}
whereas for HiHTP,
\begin{equation*}
      \rho^{\text{HiHTP}} = \left(\frac{2\delta_{3s, 2\sigma}}{1-\delta^2_{(2s, 2\sigma)}}\right)^{1/2} , \quad \tau^{\text{HiHTP}} = \frac{5.15}{1-\rho^{\text{HiHTP}}}.
\end{equation*}
\end{theorem}

The theorem's proof follows closely along the lines of the standard proofs for HTP and IHT as found, e.g.\ in Refs.~\citeWE{FouRau13, Foucart:2011, BouchotEtAl:2016}.  A detailed proof can be found in Refs.~\citeWE{wunder2019low, RothEtAl:2020:HiHTP}, respectively.

\section{Hierarchically restricted isometric measurements}
\label{sec:hiRIP}

The results of the last section make it clear that the HiRIP-property has the same role for hierarchically sparse recovery as the RIP takes on for sparse recovery. If we can prove that an operator $A$, for appropriate hi-sparsity levels $(s,\sigma)$, has the HiRIP, it is guaranteed that HiHTP can recover $x$ from the measurements $Ax$. In this chapter, we will establish the HiRIP for several families of measurement operators, using more and more specialized techniques.

\subsection{Gaussian operators} \label{sec:Gaussian}
Let us first discuss the HiRIP-properties of the arguably most well-known random  construction of a measurement operator: The Gaussian random matrix. 
A random matrix $A \in \mathbb{K}^{m\times n}$ is thereby said to be \emph{Gaussian} if the entries are i.i.d.\ distributed according to the standard normal distribution $\mathcal{N}(0,1)$.

It has become a folklore result (see, e.g., Ref.\  \citeWE[Ch.9]{FouRau13}) that if $A$ is Gaussian, the renormalized matrix $\tfrac{1}{\sqrt{m}}A$ has the $s$-RIP with high probability if 
\begin{align*}
   m \gtrsim s \log\left(\frac{n}{s}\right),
\end{align*}
where the notation $\gtrsim f(x)$  
means greater than $C\cdot f(x)$, with $C$ an unspecified universal numerical constant. 
It is therefore natural to ask how large $m$ needs to be in order for $\frac{1}{\sqrt{m}}A$ to have the $(s,\sigma)$-HiRIP. 
Since $(s,\sigma)$-sparsity is more restrictive than $s\sigma$-sparsity, we surely will not need more than 
$\text{const.} \cdot s\sigma \cdot \log\left(\tfrac{s\sigma}{nN}\right)$ 
measurements. 
But is the threshold lower for the HiRIP? And if so, how much?  

In fact, the framework of \emph{model-based compressed sensing} \citeWE{BarCevDua10}
gives us a standard route to answer this question for the Gaussian ensemble. Let us sketch this route in some detail. First, one realizes that for any normalized fixed $x \in \mathbb{K}^N$, the random vector $\tfrac{1}{\sqrt{m}}Ax$ is also Gaussian, and as such obeys the following \emph{measure concentration inequality}: 
\begin{equation*}
    \mathbb{P}\left(\,\left\vert \left\Vert \tfrac{1}{\sqrt{m}} A x \right\Vert^2 - 1 \right\vert > \delta \,\right) \leq 2 \exp\left(-cm \delta^2 \right),
\end{equation*}
where $c$ is a numerical constant. For a fixed vector $x\in \mathbb{K}^n$, $\tfrac{1}{\sqrt{m}}A$ preserves its norm with high probability.

Second, we generalize the almost isometric behaviour to hold for all vectors supported on a certain $k$-dimensional subspace $V$. 
To this end, we first establish that it suffices that the measurement operator acts almost isometrically on a so-called \emph{$\rho$-net} for the intersection of the Euclidean unit ball with $V$. 
A $\rho$-net for a set $M$ is a set $N$ with the property that for any $q \in M$, there exists a $p \in N$  with $\Vert{q-p}\Vert_2 <\rho$. 
It is not hard to construct a $\rho$-net for the normalized elements of $V$ with cardinality \citeWE{FouRau13}
\begin{equation*}
    \vert{N}\vert \leq C_{\mathrm{net}}\left(1+\tfrac{2}{\rho}\right)^k \, .
\end{equation*} 
By choosing $\rho$ suitably and applying a union bound over the $\rho$-net, we obtain for any support $S$ with $\vert{S}\vert=k$
\begin{equation} \label{eq:conc}
    \mathbb{P}\left(\, \left\vert \left\Vert \tfrac{1}{\sqrt{m}} A x \right\Vert^2 - 1 \right\vert > \delta \quad \forall x : \operatorname{supp}(x) =S \right) \leq C \lambda^k \exp\left(-\tilde{c}m \delta^2 \right),
\end{equation}
where $C$, $\lambda$ and $\tilde{c}$ are universal numerical constants. 
With \eqref{eq:conc} at our disposal, it is only one step to establish an isometry property for $\tfrac{1}{\sqrt{m}}A\in \R^{m \times N\cdot n}$ for an entire union of subspaces such as structured sparse vectors. 
For instance, in order to get the $(s,\sigma)$-HiRIP, we need to take a union bound over all $(s,\sigma)$-sparse supports $S$. There are $\binom{N}{s}\binom{n}{\sigma}^s$ such supports. Therefore
\begin{equation*}
    \mathbb{P}\left(\, \left\vert \left\Vert \tfrac{1}{\sqrt{m}} A x \right\Vert^2 - 1 \right\vert > \delta \quad \forall\ \text{$(s,\sigma)$-sparse $x$} \right) \leq C \binom{N}{s}\binom{n}{\sigma}^s \lambda^{s\sigma} \exp\left(-\tilde{c}m \delta^2\right)\, .
\end{equation*}
This probability is dominated by $\epsilon$, if
\begin{equation*} 
    m \geq \tilde{c}^{-1} \delta^{-2} \log \left(C \binom{N}{s}\binom{n}{\sigma}^s \lambda^{s\sigma}\epsilon^{-1}\right).
\end{equation*}
Using the Stirling approximation $\binom{p}{k} \sim \left(\tfrac{p}{k}\right)^k$, we obtain the more readable condition
\begin{equation*}
    m \gtrsim  \delta^{-2} \left( s \log\left(\tfrac{N}{s}\right) + s\sigma \log\left(\tfrac{n}{\sigma}\right) + \log\left(\tfrac1\epsilon\right)\right) \, .
\end{equation*}

Let us state this as a theorem.
\begin{theorem}[HiRIP for Gaussian matrices]
    Let $A\in \mathbb{K}^{m,n\cdot N}$ be random Gaussian. Then there is a universal numerical constant $C>0$ so that if
    \begin{equation} \label{eq:nbrmeasurments}
        m \geq  \tfrac{C}{\delta^2} \left(s \log\left(\tfrac{N}{s}\right) + s\sigma \log\left(\tfrac{n}{\sigma}\right) + \log\left(\tfrac1\epsilon\right)\right),
    \end{equation}
    $\tfrac{1}{\sqrt{m}}A$ has an $(s,\sigma)$-HiRIP constant $\delta_{s,\sigma}(A)\leq \delta$ with  probability as least $1-\epsilon$.
\end{theorem}

The difference of the condition \eqref{eq:nbrmeasurments} compared to one needed to establish the standard RIP,
\begin{equation} \label{eq:ripnbr}
    m \geq  \tfrac{C}{\delta^2} \left(s\sigma \log\left(\tfrac{Nn}{s\sigma}\right) +  \log\left(\tfrac 1\epsilon\right)\right)
\end{equation}
is subtle. After all, both thresholds can be written as $s \sigma$ multiplied with logarithmic terms in the dimension of surrounding space. However, for certain parameter regimes, the difference is significant. Indeed, in the scenario that $N \gg n$, \eqref{eq:nbrmeasurments} can be much smaller than \eqref{eq:ripnbr}.
This establishes that for Gaussian random matrices hierarchical thresholding algorithms are theoretically expected to have an improved sampling complexity compared to their standard counterparts. 
Also in the non-asymptotic regime one can observe an improved sample requirement in numerical simulation, see Fig.~\ref{WEfig:samplingComplexityGauss}.

\begin{figure}[tb]
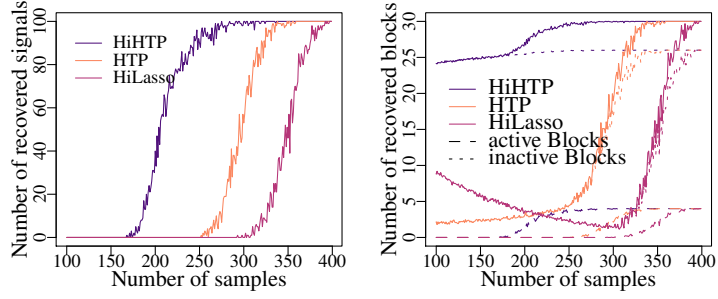

	\centering
	\input{images/plotGaussianComplete}
    \hspace{.3cm}
	\input{images/plotGaussianBlocks}
	\caption{\label{WEfig:samplingComplexityGauss}%
	Left: Number of recovered signals from $100$ noiseless Gaussian samples over the number of measurements $m$ for {HTP}, {HiLasso} and HiHTP. The signals consist of $N=30$ blocks of size $n=100$ with $s=4$ blocks having $\sigma=20$ non-vanishing real entries. Right: Number of recovered blocks over the number of measurements $m$ for {HTP} and HiHTP. The dashed and dotted lines indicate the average number of correctly recovered zero and non-zero blocks, respectively. The solid lines show the total average number of recovered blocks. The signals consist of $N=30$ blocks with $s=4$ blocks having non-vanishing real entries. A signal or block is considered recovered if it deviates from the true signal by less than $10^{-5}$ in $\ell_2$-norm. \textcopyright 2020 IEEE. Reprinted, with permission, from  Ref.~\protect\citeWE{RothEtAl:2020:HiHTP}.}
\end{figure}
    Note that the above discussion can be applied without problems to \emph{sub-Gaussian} matrices. A matrix is sub-Gaussian if the entries $a_{i,j}$ are 
    i.i.d.\ distributed according to a distribution that obeys
        $\mathbb{P}\left(\vert{a_{i,j} } \vert >t\right) \leq \alpha \exp\left(-\beta t^2\right)$
    for some $\alpha$, $\beta>0$.

\subsection{Coherence measures}
The discussion in the last section very much relies on the random nature of the measurement operator. This is a common feature of compressed sensing-related theories -- in order to obtain an optimal scaling, one practically has no choice other than to use a random construction. A possible route to still establish (non-optimal) RIP-results for non-random matrices is to take a detour via so-called \emph{coherence measures}. The simplest result is as follows \citeWE[Prop 6.2]{FouRau13}: 
If we define the \emph{mutual coherence} of a matrix with normalized columns $a_i$ as
\begin{equation*}
    \mu(A) = \sup_{i,j} \vert \langle a_i, a_j \rangle \vert,
\end{equation*}
the RIP constants obey 
\begin{equation}
    \delta_s(A) \leq (s-1)\mu(A). \label{eq:coherence}
\end{equation}

To establish analogous results for the HiRIP constants, we need to use coherence measures adapted to the block structure. 
Such measures have been introduced in Ref.\  \citeWE{SprechmannEtAl:2011} for the analysis of the HiLasso algorithm. 
To work with these coherence measures it is convenient to introduce further notation 
to refer to the blocks of a vector individually. 
To this end, we use the Kronecker product of matrices in the convention 
\begin{equation*}
  A \otimes B = \begin{pmatrix} 
    a_{1,1} B& \ldots & a_{1,N}B  \\
    \vdots & \ddots & \vdots \\
    a_{m,1} B & \ldots & a_{m,N}B
  \end{pmatrix}, 
\end{equation*}
where $a_{i,j}$ denotes the entries of $A$. 
The Kronecker product trivially also provides a Kronecker product on vectors $\K^N \times \K^n \to \K^{Nn}$ understood as $n \times 1$ and $N \times 1$ matrices, respectively. 
Using the basis $\{e_i\}_{i \in [N]}$, $(e_i)_j =\delta_{i,j}$ of $\K^N$,
we can rewrite a blocked vector $x \in \K^{Nn}$ with blocks $x_i \in \K^n$, $i \in [N]$, as the sum of products $x = (x_1^\top, x_2^\top \ldots, x_N^\top)^\top = \sum_{i \in  [N]} e_i \otimes x_i$. 
The Kronecker product exemplifies the canonical vector space isomorphism of $\K^{Nn}$ with the tensor product space $\K^N \otimes \K^n$. 
Analogously, we identify the measurement matrices $A\in \K^{m \times N\cdot n}$ with linear operators $A : \K^N \otimes \K^n \to \K^m$. 
We refer to $A_i \in  \mathbb{K}^{m \times n}$, $i\in [N]$, defined through $A_i(v) = A(e_i \otimes v)$, $v \in \K^n$, as the \emph{block-operators} of $A$. Now we introduce the specialized coherence measures. 

\begin{definition}[Sub-coherence and block-coherence]
    Let $A: \mathbb{K}^{N} \otimes  \mathbb{K}^n \to  \mathbb{K}^{m}$ with block-operators $A_i \in \K^{m \times {n}}$ and let $\{a_{i,j}\}_{j \in [n]}$ be the columns of the $i$th block-operator. We  define
    \begin{enumerate}
        \item the \emph{sub-coherence} $\nu(A)$ of $A$ as the maximal mutual coherence of the block operators, i.e.,
       \begin{equation*}
           \nu(A) = \sup_i \mu(A_i) =  \sup_i \sup_{j \neq k} \vert \langle a_{i,j},a_{i,k} \rangle \vert.
        \end{equation*}
        \item The \emph{sparse block-coherence} $\mu_{\mathrm{block}}^{\sigma\!\sigma}(A)$ of $A$ as
        \begin{equation*}
            \mu_{\mathrm{block}}^{\sigma\!\sigma}(A) = \sup_{i \neq j} \rho^{\sigma\!\sigma}(A_i^*A_j),
       \end{equation*}
        where $\rho^{\sigma\!\sigma}(B)$ denotes the \emph{$\sigma$-sparse singular value} of a matrix $B \in \mathbb{K}^{N\times N}$,
        \begin{equation*}
            \rho^{\sigma\!\sigma}(B) = \sup_{\substack{u,v  \text{ $\sigma$-sparse} \\ \Vert u \Vert = \Vert v \Vert=1}} \vert \langle u, Bv \rangle \vert.
        \end{equation*}
    \end{enumerate}
\end{definition}

Intuitively, $\nu(A)$ measures the coherence within each block, whereas $\mu_{\mathrm{block}}^{\sigma\!\sigma}(A)$ measures the coherence between the blocks. 
Note that we have used a different normalization in the definition of the sparse block-coherence compared to Ref.\ \citeWE{SprechmannEtAl:2011}. 
We can establish the following bounds on the HiRIP constants in terms of the coherence measures.

\begin{theorem}[HiRIP through coherence bound] Let $A :
\label{th:blockcoherence}\mathbb{K}^N \otimes \mathbb{K}^n \to \mathbb{K}^m$ be an operator with block-operators $A_i$ and $s\in [N]$, $\sigma \in [n]$.
It holds that
    \begin{enumerate}
        \item \quad 
            $\sup_{i} \delta_{\sigma}(A_i) \leq \delta_{1,\sigma}(A)$ \quad and\quad  
            $\mu^{\sigma\!\sigma}_{\mathrm{block}}(A) \leq  2\delta_{2,\sigma}(A)$.
        \vspace{.2\baselineskip}
        \item \quad
        $
        \delta_{s,\sigma}(A) \leq \sup_{i} \delta_{\sigma}(A_i)  + (s-1) \mu^{\sigma\!\sigma}_{\mathrm{block}}(A).
        $
    \end{enumerate}
    In addition, if all columns of the block-operators $A_i$ are normalized, then
    \begin{equation*}
         \delta_{s,\sigma}(A) \leq (\sigma-1) \nu(A) + (s-1) \mu^{\sigma\!\sigma}_{\mathrm{block}}(A)\, .
    \end{equation*}
\end{theorem}

\begin{proof}
     1. Let $j \neq k$ and $x,y\in \K^n$ be $\sigma$-sparse normalized vectors. First, we have
     \begin{equation*}
         \vert \Vert A_j x \Vert^2 - \Vert x \Vert^2 \vert = \vert \Vert A (e_j \otimes x) \Vert^2 - \Vert e_j \otimes x \Vert^2 \vert \leq \delta_{1,\sigma}(A),
    \end{equation*}
    since $e_j \otimes x$ is $(1, \sigma)$-sparse. This proves the first claim. 
    For the second claim, we use the polarization identity to find
    \begin{equation*}
        \langle A_j x, A_k y \rangle = \frac{1}{4}\sum_{\ell=0}^3 i^\ell \left\Vert A_j x + i^\ell A_k y \right\Vert^2  =  \frac{1}{4}\sum_{\ell=0}^3 i^\ell \left\Vert A(e_j \otimes  x+i^\ell e_k \otimes y) \right\Vert^2 .
    \end{equation*}
    Since $e_j \otimes x$ and $e_k \otimes y$ have disjoint block supports, $e_j \otimes x+i^\ell e_k \otimes y$ are $(2,\sigma)$-sparse for all $\ell$. Hence,
   \begin{align*}
      \left\vert\frac{1}{4}\sum_{\ell=0}^3 i^\ell \left\Vert A( e_j \otimes x +i^\ell e_k \otimes y )\right\Vert^2 -\frac{1}{4}\sum_{\ell=0}^3 i^\ell \left\Vert  e_j \otimes x +i^\ell e_k \otimes y\right\Vert^2   \right\vert \\
      \leq \delta_{2,\sigma} \cdot \frac{1}{4}\sum_{\ell=0}^3 \left\Vert  e_j \otimes x+i^\ell e_k \otimes  y \right\Vert^2.
   \end{align*}
   Now we use that $\left\Vert  e_j \otimes x+i^\ell e_k \otimes y \right\Vert^2= 2$ for all $ \ell$. This both proves that the final bound above equals $\tfrac{1}{2}\delta_{s,\sigma}$,  and that $\sum_{\ell=0}^3 i^\ell \left\Vert  e_j \otimes x +i^\ell e_k \otimes  y \right\Vert^2 =0$, yielding the claim.
   
     2.  Let $x=\sum_i e_i \otimes x_i$ be an $(s,\sigma)$-sparse and normalized signal. There exists an $S\subseteq [N]$ with $\vert{S}\vert=s$ so that $x_i=0$ for $i \notin S$. We have
     \begin{eqnarray*}
         \Vert A x \Vert^2  &=&  \sum_{i=1}^N  \Vert A_i x_i \Vert^2   + \sum_{i \neq j} \langle{A_i x_i, A_j x_j} \rangle.
     \end{eqnarray*}
     Each $x_i$ is $\sigma$-sparse and, thus, $\vert \Vert A_i x_i \Vert^2 - \Vert x_i \Vert^2 \vert\leq \delta_\sigma(A_i) \Vert x_i \Vert^2$. Taking the sum over $i$ yields
     \begin{align*}
         \left \vert \Vert x \Vert^2 - \sum_{i=1}^N  \Vert A_i x_i \Vert^2 \right\vert \leq \sup_{i} \delta_\sigma(A_i) \Vert{x}\Vert^2.
     \end{align*}
     We still need to deal with the cross-block terms.  
     Let the support of $x_i$ be denoted $S_i$, the orthogonal projection onto the space supported on $S_i$ with $P_{S_i}$, and define $V$ as the subspace with the same support as $x$. 
     Consider the operator $C: V \to V$, 
     \begin{equation*}
         y =\sum_i e_i \otimes y_i \mapsto  \sum_{i \in S } e_i \otimes P_{S_i}\left( \sum_{k \in S \setminus \{i\}
         } A_i^*A_k y_k\right).
     \end{equation*}
     We have 
     \begin{equation*}
         \sum_{i \neq j} \langle{A_i x_i, A_j x_j} \rangle  = \langle x, C x \rangle,
     \end{equation*}
     and $C$ is Hermitian. 
     The latter implies that  $\vert\langle{x,Cx}\rangle \vert \leq \lambda\Vert x \Vert^2$, where $\lambda$ is the magnitude of the largest eigenvalue of $C$. 
     To estimate $\lambda$, let $v = \sum_i e_i \otimes v_i$ be a normalized eigenvector for $C$, and $i$ such that $\Vert v_i \Vert$ is maximal.
     We have $\lambda v_i = P_{S_i} \sum_{k \in S} A_i^* A_k v_k$, and consequently
     \begin{eqnarray*}
         \lambda \Vert v_i \Vert^2   &=& \langle v_i , P_{S_i}\sum_{k \in S} A_i^* A_k v_k\rangle = \sum_{\substack{k\in S \\  k \neq i}} \langle{A_i v_i, A_k v_k}\rangle \\
         &\leq &\sum_{\substack{k\in S \\  k \neq i}} \mu_{\mathrm{block}}^{\sigma\!\sigma}(A) \Vert v_i \Vert \Vert v_k \Vert \leq (s-1) \mu_{\mathrm{block}}^{\sigma\!\sigma}(A) \Vert v_i \Vert^2.
     \end{eqnarray*}
     In the second step we have used that  $v_i =P_{S_i}v_i$,
     since $v \in V$. In the penultimate step we have used that $v_i$ and $v_k$ all are $\sigma$-sparse and that each index $k$ in the sum is different from $i$. In the final step we have used the optimality of $i$.  This proves that $\lambda \leq (s-1)\mu_{\mathrm{block}}^{\sigma\!\sigma}(A)$ and therefore the claim.
     
     Finally, the addition of the theorem follows from the claim with \eqref{eq:coherence}.
\end{proof}

The above result can be applied to construct a large family of operators that have suitably small HiRIP constants without exhibiting RIP in this regime. 
Consider $N$ pairwise orthogonal, $p$-dimensional subspaces of $\mathbb{K}^m$, and $E_i: \mathbb{K}^{p}\to \mathbb{K}^m$ isometric embeddings onto them. 
Let further $C \in \mathbb{K}^{p\times n}$ be a fixed matrix with $\delta_\sigma(C) = \delta<1$. 
We consider the operator 
\begin{equation*}
    A : \mathbb{K}^N \otimes \mathbb{K}^n \to \mathbb{K}^m, \ x \mapsto \sum_{i=1}^N E_i C x_i.
\end{equation*}
The block operators of $A$ are given by $E_iC$, $i \in [N]$ and each of them is compressively encoding $\K^n$ into one of the mutually orthogonal subspaces. 
Due  to the fact that the $E_i$ are isometric, $\delta_\sigma(A_i) = \delta_\sigma(C)$ for each $i$. 
The pairwise orthogonality of the subspaces imply that $A_i^*A_j=0$ for $i \neq j$, so that $\mu_{\mathrm{block}}^{\sigma\!\sigma}(A)=0$. 
Theorem~\ref{th:blockcoherence} then implies that $\delta_{s,\sigma}(A) \leq \delta(C)$ for any $s$. 

The above construction will generically not result in an operator with small $\delta_{s\sigma}(A)$. 
To this end, suppose that $s\sigma \leq n$ and $p \leq n-s\sigma$.
Then, there exists an $s\sigma$-sparse $w \in \mathbb{K}^n$ with $Cw =0$. 
Now the vector $x=(w, 0, \dots, 0)$ is $s\sigma$-sparse, but 
$
    \left\vert \Vert Ax \Vert^2 - \Vert x \Vert^2 \right\vert  = \left\vert \Vert 0 \Vert^2 - \Vert x \Vert^2 \right\vert =  \Vert x \Vert^2
$.
We conclude that $\delta_{s\sigma}(A)\geq 1$. %

A disadvantage of this construction is that necessarily $m \geq Np \geq N\sigma$. 
This is a considerably worse scaling than we found for Gaussian random matrices, which exhibit the HiRIP for $m \gtrsim s\sigma$ up to log-factors. 
The scaling in $N$ as opposed to the sparsity parameter on the block-level $s$ arises 
from the encoding into mutually orthogonal subspaces.
The idea of `mixing' block operators can, however, be driven a lot further to avoid this overhead as we will see in the next section.

\subsection{Hierarchical measurement operators} \label{sec:HiMesOp}
As we saw above, a measurement operator on $\mathbb{K}^N \otimes \mathbb{K}^n$ can always be thought of as a mixture of block operators, say
\begin{equation*}
    B(x) = \sum_{i=1}^N B_i x_i.
\end{equation*}
The inequalities in Theorem~\ref{th:blockcoherence}, part 1 imply that in order for $B$ to have a small HiRIP constant, we need each block operator to be well-conditioned, and in addition that the blocks are incoherent. What can we do when they are not? 

Assume that instead of just observing $Bx$, we are allowed to sample a few different linear combinations of the vectors $B_ix_i$,
\begin{align*}
    y =  \left(\sum_{i=1}^N a_{i,j}B_ix_i\right)_{j \in [M]} = \sum_{i =1}^N a_i \otimes B_ix_i,
\end{align*}
with $a_i = (a_{j,i})_{j \in [M]} \in \K^M$.
Can this make recovery easier? 
Let us define such measurement operators that act hierarchically on the block structure of the vectors as \emph{hierarchical measurement operators}. 
\begin{definition}[Hierarchical measurement operators]
    Let $A\in \mathbb{K}^{M,N}$ and $B_i \in \mathbb{K}^{m,n}$, $i=1, \dots, N$, be given and denote the $i$th column of $A$ by $a_i$. We call the operator
    \begin{equation*}
        \mathcal{H}: \mathbb{K}^N \otimes \mathbb{K}^n \to \mathbb{K}^M \otimes \mathbb{K}^m, \qquad x \mapsto \sum_{i =1}^N a_i \otimes B_i x_i
    \end{equation*}
    the \emph{hierarchical measurement operator defined by $A$ and $(B_i)_{i \in [N]}$}.
\end{definition}

The structure and naming of  hierarchical operators makes it easy to believe that they are an excellent fit for hierarchically sparse recovery. 
They are, however, by no means only of academic interest. We will discuss this more thoroughly in Sec.~\ref{section:applications}. 
For now, the practical interest might already become apparent by noting that 
an important special case of hierarchical measurement operators is the following:
In the case of $B_i = B $ being equal, the hierarchical operator is the same as the \emph{Kronecker product} $A\otimes B$ of the matrices $A$ and $B$. 
How do the hierarchical isometry constants of $\mathcal{H}$ relate to the ones of $A$ and the $B_i$s? In order to discuss this question, we begin by proving the following lemma.

\begin{lemma}
[RIP implies nuclear norm isometry]
    Let 
    $X \in \mathbb{K}^{N \times N}$ have the property that for some sets $S, \overline{S}$ of cardinality $s$, $X_{i,j}=0$ if either $i \notin S$ or $j \notin \overline{S}$.\label{lem:schatten}
    \begin{enumerate}
        \item If $X$ is positive definite Hermitian, which in particular implies $S=\overline{S}$,
        \begin{equation*}
            \left\vert \langle A^*A, X \rangle  - \Vert X \Vert_* \right\vert \leq \delta_s(A) \Vert X \Vert_*\, .
        \end{equation*}
        \item If $S$ and $\overline{S}$ are disjoint,
        \begin{equation*}
            \vert \langle A^*A,X \rangle \vert \leq \delta_{2s}(A) \Vert X \Vert_* \, .
        \end{equation*}
 
    \end{enumerate}
    Here,
        $\Vert X \Vert_*$ denotes the nuclear norm, also known as the trace norm, of $X$, i.e. the sum of its singular values.
\end{lemma}
\begin{proof}
     Consider a singular value decomposition of $X$, $X= \sum_{i=1}^N \sigma_i v_i  u_i^*$.
     We have
     $%
         \langle A^* A , X \rangle = \sum_{i=1}^N \sigma_i \langle{Au_i, A v_i}\rangle%
    $.
     Due to the assumption, for all $i$ with $\sigma_i \neq 0$,  $\operatorname{supp}(v_i) \subset S$ and $\operatorname{supp}(v_i) \subset \overline{S}$. 

1. For $X$ positive-definite, the $\sigma_i$ are the eigenvalues of $X$, and $u_i=v_i$. Since each $u_i$ is $s$-sparse, it holds that
        \begin{equation*}
            \left\vert \langle A^*A, X \rangle  - \Vert X \Vert_* \right\vert  \leq \sum_{i=1}^N \sigma_i \left\vert \langle{Au_i, A u_i}\rangle -1 \right\vert \leq  \sum_{i=1}^N \sigma_i \cdot \delta_s(A) = \delta_s(A) \Vert X \Vert_*.
        \end{equation*}
        
    2. Ref.~\citeWE[Prop. 6.3]{FouRau13} states that since the supports of $u_i$ and $v_i$ are disjoint, we have $\vert \langle A u_i, Av_i \rangle \vert \leq \delta_{2s}(A)$. This in turn implies
        \begin{equation*}
            \vert \langle A^*A, X \rangle \vert  \leq \sum_{i=1}^N \sigma_i \vert \langle{Au_i, A v_i}\rangle \vert \leq  \sum_{i=1}^N \sigma_i \delta_{2s}(A) = \delta_{2s}(A) \Vert X \Vert_* \,. 
        \end{equation*}
        
        { \flushright $\qed$}
\end{proof}

We now prove that $\mathcal H$ inherits the HiRIP from the RIP of its constituent matrices, in that $\delta_{s,\sigma}(\mathcal{H})$ can be bounded in terms of  $\delta_s(A)$ and the constants $\delta_\sigma(B_i)$.

\begin{theorem}[{Hierarchically inherited HiRIP}] 
\label{th:hihiRIP}
    Let $\mathcal{H}$ be the hierarchical operator defined by $A$ and $(B_i)_{i\in [N]}$. We have for $s,\sigma$ arbitrary
    \begin{equation*}
        \delta_{s,\sigma}(\mathcal{H}) \leq \delta_{s}(A) + \sup_{i}\delta_{\sigma}(B_i) + \delta_s(A) \cdot \sup_{i} \delta_\sigma(B_i).
    \end{equation*}
\end{theorem}
\begin{proof}
     Let $x$ be normalized and $(s,\sigma)$-sparse, and $S$ such that $a_i =0$ for $i \notin S$. We have
     \begin{equation*}
         \Vert \mathcal{H}(x) \Vert = \sum_{i,j=1}^N \langle{a_i \otimes (B_ix_i), a_j \otimes(B_j x_j) }\rangle = \sum_{i,j=1}^N \langle a_i, a_j \rangle \langle B_i x_i, B_j x_j\rangle = \langle A^*A, G \rangle, \label{eq:HSrepresentation}
     \end{equation*}
     where $G \in \K^{N\times N}$ denotes the matrix with non-vanishing entries $G_{i,j} = \langle B_i x_i, B_j x_j \rangle$ for $i \in S$ and $j \in S$.
     By Lemma~\ref{lem:schatten}, part 1,  
     \begin{equation} \label{eq:estimateone}
         \vert\langle{A^*A,G}\rangle - \Vert G \Vert_* \vert \leq \delta_s(A) \Vert G \Vert_*.
     \end{equation}
     It remains to estimate $\Vert G \Vert_*$. In order to do this, consider the operator
     $%
         M : \mathbb{K}^{{|}S{|}} \to \mathbb{K}^{m}$%
         , $c \mapsto \sum_{i\in S} c_i B_i x_i$%
     .
     By construction, $G=M^*M$, and therefore, $\Vert{G}\Vert_* = \Vert M \Vert^2= \sum_{i\in S} \Vert B_i x_i \Vert^2$, where $\Vert \, \cdot \, \Vert$ here refers to the Frobenius norm. Consequently,
     \begin{equation} \label{eq:estimatetwo}
         \left\vert \Vert{G}\Vert_* - \Vert x \Vert^2 \right\vert \leq \sum_{i\in S} \left\vert \Vert B_i x_i \Vert^2 - \Vert x_i \Vert^2 \right\vert \leq \sum_{i\in S} \delta_{\sigma}(B_i) \Vert x_i \Vert^2.
     \end{equation}
     Combining \eqref{eq:estimateone} and \eqref{eq:estimatetwo}, we obtain
     \begin{eqnarray*}
         \left\vert\langle{A^*A,G}\rangle - \Vert x \Vert^2 \right\vert &\leq& \left\vert\langle{A^*A,G}\rangle - \Vert G \Vert_* \right\vert + %
         \left\vert \Vert{G}\Vert_* - \Vert x \Vert^2 \right\vert 
         \\
         &\leq& \delta_s(A) \left( 1+\sup_{i}\delta_{\sigma}(B_i) \right) \Vert x\Vert^2+ \sup_i \delta_{\sigma}(B_i)\Vert x \Vert^2,
     \end{eqnarray*}
     which proves the claim.
\end{proof}

The theorem shows that hierarchical operators are a rich class of operators which much more often have the HiRIP than the RIP.  
To make this precise, we take a look at the special case of Kronecker products $A\otimes B$.
Theorem~\ref{th:hihiRIP} implies that if $\delta_s(A)$ and $\delta_{\sigma}(B)$ are small, $\delta_{s,\sigma}(A \otimes B)$ is also small. This is in stark contrast to the RIP of Kronecker products. 
Indeed, Ref.\  \citeWE{jokar:2009sparse} derived that
\begin{align*}
        \delta_{s}(A \otimes B) \geq \max(\delta_s(A), \delta_s(B)).
\end{align*}
That is, in order for $A\otimes B$ to have the $s$-RIP (nota bene, not the $s\sigma$-RIP), \emph{both} $A$ and $B$ must have it. 
This obstacle leads to demanding performance bounds in applications \citeWE{ShabaraKoksalEkici:2021}. 

The total number of measurements measured by a hierarchical operator is equal to $mM$. 
Together with the classical results on the RIP of  Gaussian operators, the theorem implies that by choosing $A$ and $B$ Gaussian we can hence build hierarchical operators having the $(s,\sigma)$-HiRIP using only
    \begin{align*}
        \mathrm{const} \cdot s\sigma \log \left(\frac{n}{\sigma}\right) \log\left(\frac{N}{s}\right)
    \end{align*}
    many measurements. 
    This scaling is up to log-factors identical to the result Eq.\ \eqref{eq:nbrmeasurments} we established for fully Gaussian matrices. 
    This is noteworthy, since while fully Gaussian matrix consists of $MN \cdot mn$ independent parameters, a Kronecker product $A\otimes B$ only has $MN + mn$. 
    This constitutes a considerable de-randomization of the measurements, which can be e.g.\ exploited to reduce the storage complexity or to speed up calculations. 
    We refer to Refs.~\citeWE{roth2018hierarchical, RothEtAl:2020:HiHTP} for an extended discussion and an alternative direct proof of HiRIP for Kronecker product measurements.

Theorem~\ref{th:hihiRIP} tells us that operators with small RIP constants can be combined to obtain an operator with a small HiRIP constant. 
We now take a look at the contrary question: To what extend 
are small RIP constants of the constituent operators required to 
bound the HiRIP constants of the hierarchical measurement operator?
 
In order to get a simple formulation of our first result, let us first note that there is an ambiguity in the definition of a hierarchical measurement operator. 
We can always simultaneously rescale $a_i$ and $B_i$ since $a_i \otimes B_i = (\lambda a_i) \otimes (\lambda^{-1}B_i)$. 
We may thus w.l.o.g.\ assume that $\Vert a_i \Vert =1$ for all $i$. 
Under this assumption, a small $(s,\sigma)$-HiRIP constant of  $\mathcal{H}$ indeed implies small $\sigma$-RIP constants of all $B_i$.

\begin{proposition}
[$\sigma$-RIP bound from $(s,\sigma)$-HiRIP]
Let $\mathcal{H}$ be a hierarchical measurement operator given by $A$ and $(B_i)_{i\in [N]}$. Assume that the columns of $A$ fulfil $\Vert a_i \Vert =1$ for all $i$. 
Then, it holds that
    \begin{equation*}
        \sup_{i} \delta_{\sigma}(B_i) \leq \delta_{1,\sigma}(\mathcal{H}).
    \end{equation*}
\end{proposition}
\begin{proof}
    The $i$th block-operator $\mathcal{H}_i$ of $\mathcal{H}$ is given by $a_i \otimes B_i \in \mathbb{K}^{Mm \times N}$. The normalization implies that
    $
        \Vert (a_i \otimes B_i)x \Vert^2 = \Vert a_i \Vert^2 \cdot \Vert B_i x \Vert^2 = \Vert B_i x \Vert^2
    $
    for each $x \in \mathbb{K}^N$. Thus,  $\delta_{\sigma}(B_i) = \delta_{\sigma}(\mathcal{H}_i)$, and the result follows from Theorem~\ref{th:blockcoherence}, part 1.
\end{proof}

The above result in essence states that for $\mathcal{H}$ to have the $(s,\sigma)$-HiRIP, it is necessary that all $B_i$ have the corresponding $\sigma$-RIP. 
Intriguingly, for the RIP requirement of $A$, the situation is very different. 
Indeed, if the $B_i$ are mapping into incoherent subspaces, $A$ does not need to have the RIP. The precise result is as follows.

\label{th:incoherent_blocks}
\begin{theorem}[HiRIP with block incoherence] 
    For a family
    $(B_i)_{i\in [N]}$, define the operator 
    \begin{equation*}
        \mathcal{B} : \mathbb{K}^N \otimes \mathbb{K}^n \to \mathbb{K}^m, \quad x \mapsto \sum_{i=1}^N B_i x_i\, .
    \end{equation*}
    Let $A \in \mathbb{K}^{M\times N}$ and natural numbers $s$, $\sigma$ and $t$ be given. 
    The hierarchical operator $\mathcal{H}$ given by $A$ and $(B_i)_{i \in [N]}$ fulfils
    \begin{equation*}
        \delta_{ts,\sigma}(\mathcal{H}) \leq \sup_{i} \delta_{\sigma}(B_i) + \delta_s(A) \cdot \sup_{i} \delta_{\sigma}(B_i) + t\sqrt{s} \cdot \delta_{2s}(A) \cdot \mu_{\mathrm{block}}^{(2\sigma,2\sigma)}(\mathcal{B}) \, .
    \end{equation*}
\end{theorem}

\begin{proof}
    Let $x  = \sum_i e_i \otimes x_i$ be a $(ts,\sigma)$-sparse, normalized vector, and $S  \subset [N]$ be such that $x_i=0$ for $i  \notin S$. We may subdivide $S$ into $t$ disjoint sets $S_1, \ldots, S_t$ with cardinality $s$ each. For each pair $(k,\ell) \in [t] \times [t]$, we define a matrix $G^{k,\ell}\in \mathbb{K}^{N \times N}$ with non-vanishing entries 
    $%
        G^{k,\ell}_{i,j} = \langle{B_i x_i, B_j x_j}\rangle %
    $ 
    for $i\in S_k$ and $j\in S_\ell$. 
    We may use the same reasoning as in the proof of Theorem~\ref{th:hihiRIP} to argue that
    \begin{equation*}
        \Vert \mathcal{H}(x)\Vert^2 = \sum_{k=1}^ \langle{A^*A,G^{k,k}}\rangle + \sum_{k \neq \ell} \langle A^*A, G^{k,\ell}\rangle \, .
    \end{equation*}
    Now, each matrix $G^{k,\ell}$ fulfills the assumption of Lemma~\ref{lem:schatten}, part 1 for $k =\ell$, and Lemma~\ref{lem:schatten}, part 2 for $k \neq \ell$. 
    Hence,
    \begin{equation*}
        \left\vert \Vert \mathcal{H}(x)\Vert^2  - \sum_{k=1}^t \left\Vert G^{k,k}\right\Vert_* \right\vert \leq \delta_s(A) \cdot  \sum_{k=1}^N \left\Vert G^{k,k}\right\Vert_* + \delta_{2s}(A) \cdot \sum_{k\neq \ell} \left\Vert G^{k,\ell}\right\Vert_*
    \end{equation*}
    Still in analogy to the proof of Theorem~\ref{th:hihiRIP}, we find that $\left\vert\, \left\Vert{G^{k,k}}\right\Vert_* - \Vert{x_k}\Vert^2 \,\right\vert \leq \sup_{i} \delta_{\sigma}(B_i) \Vert{x_k}\Vert^2$, and consequently
    \begin{align*}
        \left\vert \mathcal{H}(x) - \Vert x \Vert^2 \right\vert \leq \delta_s(A)\left(1+\sup_{i}\delta_\sigma(B_i)\right)  + \delta_{2s}(A) \cdot \sum_{k\neq \ell} \left\Vert G^{k,\ell}\right\Vert_*\,.
    \end{align*}
    It remains to bound the terms with $k \neq l$. First, let us note that, since $G^{k,\ell}$ has rank at most $s$, $\Vert G^{k,\ell}\Vert_* \leq \sqrt{s} \Vert G^{k,\ell} \Vert$. We now use the definition of the intra-block coherence to argue that
    \begin{equation*}
        \left\Vert G^{k,\ell} \right\Vert = \sqrt{\sum_{i \in S_k, j \in S_\ell} \vert \langle B_i x_i, B_j x_j \rangle \vert^2} \leq \mu_{\mathrm{block}}^{(2\sigma,2\sigma)} \sqrt{\sum_{i \in S_k, j \in S_\ell}  \Vert x_i \Vert^2 \cdot \Vert x_j  \Vert^2}\,.
    \end{equation*}
Finally with
\begin{eqnarray*}
    \sum_{k\neq \ell} \sqrt{\sum_{i \in S_k,}  \Vert x_i \Vert^2} \cdot \sqrt{ \sum_{j \in S_\ell} \Vert x_j  \Vert^2} \leq \left(\sum_{k} \sqrt{\sum_{i \in S_k}  \Vert x_i \Vert^2} \right)^2 \leq t \Vert{x}\Vert^2\,,
\end{eqnarray*}
where we have used the Cauchy-Schwarz inequality in the final step, the claim follows.
\end{proof}

Note that the above result shows that $A$ does not need to have the $ts$-RIP in order for the hierarchical operator to exhibit the corresponding HiRIP. 
We may in particular choose $t=N/s$ and obtain an operator that acts isometrically on any vector with sparse blocks. 
In terms of sample complexity, the above result is still a bit opaque. 
By making a particular choice of $t$ and using the methods of Gaussian random matrices discussed in Sec.~\ref{sec:Gaussian}, one can derive the following result (see
Ref.\ \citeWE{gross2021hierarchical} for a proof).

\begin{proposition}[Sample complexity]
   Let $({B}_i)_i$ and $\mathcal{B}$  be as in Theorem~\ref{th:incoherent_blocks}.
   Assume that
   \begin{equation*}
   \left(t\mu_{\mathrm{block}}^{(2\sigma,2\sigma)}(\mathcal{B})\right)^2\leq \frac N {\log(N)}
   \end{equation*}
   and choose $A\in \mathbb{K}^{M\times N}$ as a Gaussian matrix. Let $\delta,\epsilon>0$. Provided that 
    \begin{align*}
        M \geq C\left( \left(t\mu_{\mathrm{block}}^{(2\sigma,2\sigma)}(\mathcal{B})\right)^2 \cdot \tfrac{1}{\delta^2}\log\left(\frac{N\left(1+\sup_i \delta_{\sigma(B_i)}\right)^2}{\left(t\mu_{\mathrm{block}}^{(2\sigma,2\sigma)}(\mathcal{B})\right)^2}\right) + \log\left(\tfrac1\epsilon\right)\right),
    \end{align*}
    where $C$ is a universal numerical constant, the hierarchical measurement  operator $\mathcal{H}$ defined by $A$ and $(B_i)_{i\in N}$ obeys
    \begin{align*}
        \delta_{t,\sigma}(\mathcal{H}) \leq \delta +\sup_i \delta_{\sigma}(B_i)
    \end{align*}
    with a probability at least $1-\epsilon$.
\end{proposition}
This proposition shows that if $\mu_{\mathrm{block}}^{(2\sigma,2\sigma)}(\mathcal{B})$ is small enough, the number of `Gaussian linear combinations' we take with $A$ does not have to grow linearly in $t$ in order to establish a $(t, \sigma)$-RIP -- instead, only $(t\mu_{\mathrm{block}}^{(2\sigma,2\sigma)}(\mathcal{B}))^2$ is needed. 

The square dependence here on $(t\mu^{(2\sigma,2\sigma)})$ is of course inferior compared to the linear dependence of the sparsity we can achieve with the help of Theorem~\ref{th:hihiRIP}. It is unclear whether this is merely an artefact of the proof.

This results end our discussion of the hierarchical operators, and with that our theoretical results on hierarchical restricted isometry properties. 

\section{Sparse de-mixing of low-rank matrices} \label{sec:sparseLowRank}

Generally, hierarchically sparse vectors arise from recursively assuming nested groupings of the vector entries to be sparsely non-vanishing.
Another generalization of hierarchically structured vectors arise when we replace the sparsity assumption with another structure assumption such as a low-rank when suitably reshape.
One of the simplest of such examples is the de-mixing of a sparse sum of low-rank matrices from linear measurements. 
For $i \in [N]$, let $\mathcal A_i: \K^{n \times n} \to \K^m$ be linear maps  and $\rho_i \in \K^{n \times n}$ be matrices of rank at most $r$. 
The problem of \emph{de-mixing low-rank matrices} is to reconstruct the matrices $\rho_i$ given data of the form 
\begin{equation*}
    y = \sum^N_{i=1} \mathcal A_i(\rho_i)\, .
\end{equation*}
A further structure assumption might be that out of the $N$ matrices $\rho_i$ actually only a number of $s$ are non-vanishing, giving rise to the problem of de-mixing a sparse sum. 
We can straight-forwardly cast the problem as the reconstruction problem of a hierarchically structured vector.  
To this end, we set $X = \sum_{i=1}^N e_i \otimes \rho_i$. We can regard $X$ as a `vector' in $\K^{Nn \times n}$ of matrix-valued blocks of rank-$r$ and at most $s$ vanishing blocks.

Compared to $(s,\sigma)$-sparse vectors, we have replaced the non-vanishing $\sigma$-sparse blocks by low-rank matrices.
The de-mixing problem of a sparse sum of low-rank matrices then is the task to reconstruct such a hierarchically (block) sparse, (block-wise) low-rank vector $X$ from linear measurements.

The principle strategy of hierarchical hard-thresholding of Sec.~\ref{sec:hiRecoveryAlgs} carries over to hierarchically sparse, low-rank vectors. 
The projection onto the set of rank-$r$ matrices is given by the hard-thresholding of the singular values. 
Let $\rho \in \K^{n \times n}$ have singular value decomposition $U \operatorname{diag}(\Sigma) V^*$ with a vector of singular values $\Sigma \in \K^n$. 
We define 
\begin{equation*}
    P_{r} (\rho) = U \operatorname{diag}(\T_r(\Sigma)) V^*\, . 
\end{equation*}
Basically, replacing the application of $\T_\sigma$ in the hierarchically thresholding Alg.~\ref{WEalg:hierarchical_hard_thresholding} yields a projection onto hierarchically sparse, low-rank vectors which we will refer to as $\bar\T_{s, r}$.

\begin{algorithm}[tb]
	\caption{SDT-algorithm}\label{WEalg:SDTalgCompact}
	\SetAlgoLined
    \SetKwInOut{Input}{input}\SetKwInOut{Output}{output}
    \SetKwInOut{Init}{initialize}\SetKwFunction{Break}{break}
    \DontPrintSemicolon
		\Input{Data $y$, measurement $\mathcal A$, sparsity $s$ and rank $r$ of signal}
		\Init{$X^0=0$.}
		\Repeat{stopping criterion is met at $l=l^\ast$}{
			Calculate step-widths $\mu^l$ \;
			$X^{l+1} 
					= \bar\T_{s,r}
					\left(
						X^l + \operatorname{diag}(\mu^l) P_{\mathcal{T}_{X^l}}
						\left(
							\mathcal{A}^*
							\left(
								y - \mathcal{A}(X^l)
							\right)
						\right)
					\right)$ \;
		}
		\Output{Recovered signal $X^{l^\ast}$}
\end{algorithm}

Modifying the projective gradient-descent of the HiIHT algorithm with this projection yields the so-called \emph{sparse de-mixing thresholding} (SDT) algorithms, Alg.~\ref{WEalg:SDTalgCompact} \citeWE{RothEtAl:2020:Semidevicedependent}. 
In contrast to the structure of a union of subspaces of sparse vectors, the set of rank $r$ matrices constitutes an embedded differential manifold in the linear vector space of all matrices. 
The geometrical structure can be exploited in iterative hard-thresholding algorithms by projecting the gradient of the embedding space in the descent step onto the tangent space of the manifold at the current iterate \citeWE{WeiEtAl:2016, AbsilSepulchre:2009,Vandereycken:2013}. 
At point $\rho$, the tangent space of the manifold of rank-$r$ matrices is the linear span of the set of matrices that have the same row or column space as $\rho$ \citeWE{AbsilSepulchre:2009}. 
For a hierarchically sparse, low-rank vector $X = \sum_{i=1}^N e_i \otimes \rho_i$, 
we use the projection onto the tangent space for each block. 
We denote by $P_{V_i}$ and $P_{U_i}$ be the projection onto the row and column space of $\rho_i$, respectively. 
For $\rho_i$ vanishing we set the projections to be the identity. 
We define $P_{\mathcal T_X}: \K^{Nn \times n} \to \K^{Nn \times n}$ as $G = \sum_{i=1}^N e_i \otimes g_i \mapsto  \sum_{i=1}^N e_i \otimes [g_i - (\operatorname{Id} - {P_{U_i}})g_i(\operatorname{Id} - {P_{V_i}})]$. 
The particularity of the SDT algorithm is that we allow for a different step-size for each matrix block.
We refer to Ref.~\citeWE{RothEtAl:2020:Semidevicedependent} for more details on the algorithm and Ref.~\citeWE{BSTrepo} for an implementation. 
The SDT algorithm without the sparse-thresholding operation to determine the block support coincides with algorithm proposed in Ref.~\citeWE{StrohmerWei:2017}.

\begin{figure}[tb]
    \centering
	\includegraphics[width=.7\textwidth]{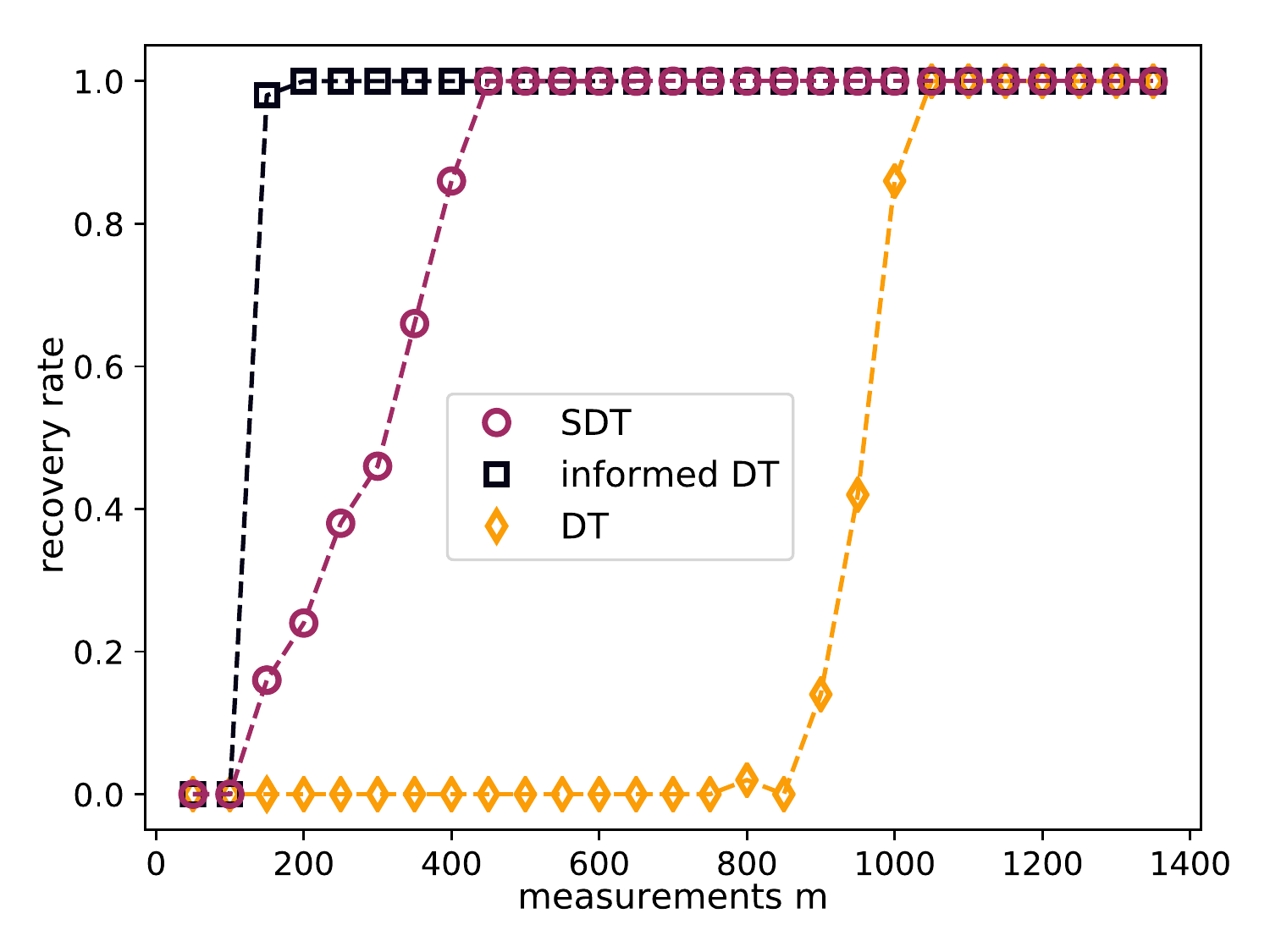}
	\label{WEfig:numericsGauss} 
	\caption{
	The figure (taken from Ref.~\protect\citeWE{RothEtAl:2020:Semidevicedependent}) displays the recovery rate for the SDT in different variants for different values of $m$ for random Gaussian measurements. 
	\emph{DT} refers to the SDT algorithms without the sparsity constraint (SD), and \emph{informed DT} to the SDT algorithm restricted to the correct support.
	An instance is considered successfuly recovery if it deviates from the true signal by less than $10^{-3}$ in Frobenius norm. 
	Each point is averaged over $50$ iterations and signal instances with $r=1$, $n=16$, $N=10$ and $s=3$.
	One observes nearly coinciding recovery performances for the informed DT and the SDT algorithm. In
	 comparison, the DT algorithm requires significantly more samples for recovery. 
	}
\end{figure}

Following the blue-print of model-based compressed sensing one can also establish a recovery guarantee based on a RIP condition custom-tailored to the hierarchical structure at hand. 
For random Gaussian measurement ensembles this gives rise to a sampling complexity of 
\begin{equation*}
    \delta^{-2} [s \log(N/s) + (2n + 1) r s \log \tfrac1\delta] 
\end{equation*}
to guarantee the correct recovery of $X \in \K^{Nn \times n}$ with at most $s$ non-vanishing blocks of rank $r$ \citeWE[Theorem 6]{RothEtAl:2020:Semidevicedependent}. 
Many results derived in Sec.~\ref{sec:hiRIP} that establish the HiRIP for hierarchically sparse vectors for different measurement ensembles, can be  generalized to hierarchically sparse, low-rank vectors.  
This allows one to guarantee recovery by the SDT algorithm for a large class of measurement ensembles.

Compared to an algorithm that does not exploit the sparsity of the de-mixing problem, the SDT algorithm can exhibit a significant improvement in the sampling complexity in relevant parameter regimes, Fig.~\ref{WEfig:numericsGauss}. 

Hierarchically sparse, low-rank vectors certainly constitute another important class of hierarchically structured signals as it encodes the de-mixing problem of a sparse sum of low-rank matrices. 
The theme of hierarchically combining low-rank and sparse structure assumptions in nested grouping of entries gives rise to a plethora of structures all of which can be efficiently reconstructed using recursive combinations of the hierarchical thresholding method introduced above.

\section{Selected applications} \label{section:applications}

\subsection{Channel estimation in mobile communication} \label{subsec:MIMO} 

In mobile communication, a lot of users are simultaneously communicating  with a base station through electro-magnetic waves. Let us model the messages a user wants to transmit with a sequence $c \in \K^n$. To send this message, the user must first translate the message to a wave. A popular scheme for this is so-called \emph{OFDM (Orthogonal Frequency-Division Multiplexing)}. 
This scheme can be imagined as each $c_k$ giving rise to a complex exponential, a so-called tone, $b(\omega) = [1, e^{-i\omega t_1}, \dots e^{-i\omega t_{n-1}}] \in \K^{1 \times n}$, where $\omega$ is the frequency and $t_1, \dots, t_{n-1}$ are some discretization times. 
In OFDM, a fixed grid of the form $\omega_k=2\pi k \overline{\omega}$, $k\in [n]$ is used where $\overline{\omega}$ is the normalized frequency. Mathematically, this corresponds to applying the discrete Fourier-transform to $c$.

As the electromagnetic waves travel from the user to the base station, they scatter on random features, e.g. buildings and trees, in the environment. This scattering causes random phase and amplitude shifts, modelled by so-called complex gains $\rho_p$. It also means that a single transmission results in several incoming wave-fronts, each with a different angle of arrival. 
This situation can be utilized if the the base station has several antennas arranged in an array: When the wave-front arrives at the antenna array, the wave-front travels slightly different distances before arriving at each antenna, i.e. if a `$1$' arrives at antenna $0$, antenna $k$ will receive an `$a_k(\theta)$', where $\theta$ denotes the angle of the wavefront. Here, $a=[a_0, \dots, a_{n-1}]: [-\pi, \pi] \to \K^{1 \times n}$ is a function, often referred to as the \emph{antenna manifold} in the communication literature. For the  popular \emph{uniform linear array (ULA)}, in which the antennae are placed at a uniform separation $d$ along a straight line, the antenna manifold is 
after a change of variables $u= d\sin(\theta)$ given by
\begin{align*}
    a(u) = \begin{bmatrix} 1, e^{2\pi d iu}, e^{4\pi d iu}, \dots ,e^{2(n-1)\pi d iu} \end{bmatrix} \,. 
\end{align*}
The parameter $u$ actually takes on values in the entirety of $[-d,d]$, but let us for now assume that it lies on some grid $\{ -\tfrac{d}{2N}, \dots,  \tfrac{d}{2N} \}$.

Combining these two models, we see that for a specific user, all transmitted signals result in a collective measurement of the form $
    \sum_{\ell=1}^L \rho_p a(u_p)^*\langle b(\omega_p)^*,c\rangle$
where $(\omega_p,u_p)$ is given by the delay and angle of the $k$th wave-front. The communication is thus characterized by the \emph{channel matrix} \citeWE{chen:2016pilot}
\begin{align*}
    H= \sum_{p=1}^L \rho_p a(u_p)^* b(\omega_p)   \in \K^{N\times n} \,.
\end{align*}
Once we know $H$, the base station can easily decode any number of sent messages. Note that as long as the environment and the position of the user does not change drastically, $H$ is expected to stay roughly constant.

Now suppose that we are only given a low dimensional sample of $H$. To be concrete, define sub-sampling operators $P_u \in \K^{M \times N}, P_\omega\in \K^{m \times n}$ in angle and delay, and assume that we only observe $P_u H P_\omega^\top$. Can we still recover the entire matrix? To do this, we may utilize that, according to the above discussion, it has a sparse representation in the delay-angle domain. Indeed, defining $A = [e^{2k \pi iu_j}]_{k, j \in [N]} \in \K^{N \times N}$ and $B= [e^{-it_k\omega_\ell}]_{k, \ell \in[n]} \in \K^{n \times n}$, we get
\begin{align*}
    P_uHP_\omega^\top= P_u A \left(\sum_{p\in [L]} \rho_p e_{u_{j_p}}\otimes e_{{\omega_{\ell_p}}} \right) B^*P_\omega^\top = (P_u A \otimes P_\omega B) X,
\end{align*}
with $X= \sum_{p=1}^L \rho_p e_{u_{j_p}}\otimes e_{\omega_{\ell_p}}$. 
Note that $X$ is not only sparse, but hierarchically sparse: only a few angle blocks are active, and for each such angle, only a few delays $\omega_k$ are utilized and vice versa. 
In fact, it is a reasonable assumption that the angles for the $L$ paths are distinct, leading to a $(1,L)$-sparse ground truth. 
We further observe that sampled $H$ is a Kronecker product measurement of $X$, where the terms of the Kronecker product are sub-sampled Fourier matrices. Thus, the results of Sec.~\ref{sec:HiMesOp} imply that the recovery indeed is possible and provide an explicit sampling complexity.

\begin{figure}[tb]
    \centering
    \includegraphics[width=.7\textwidth]{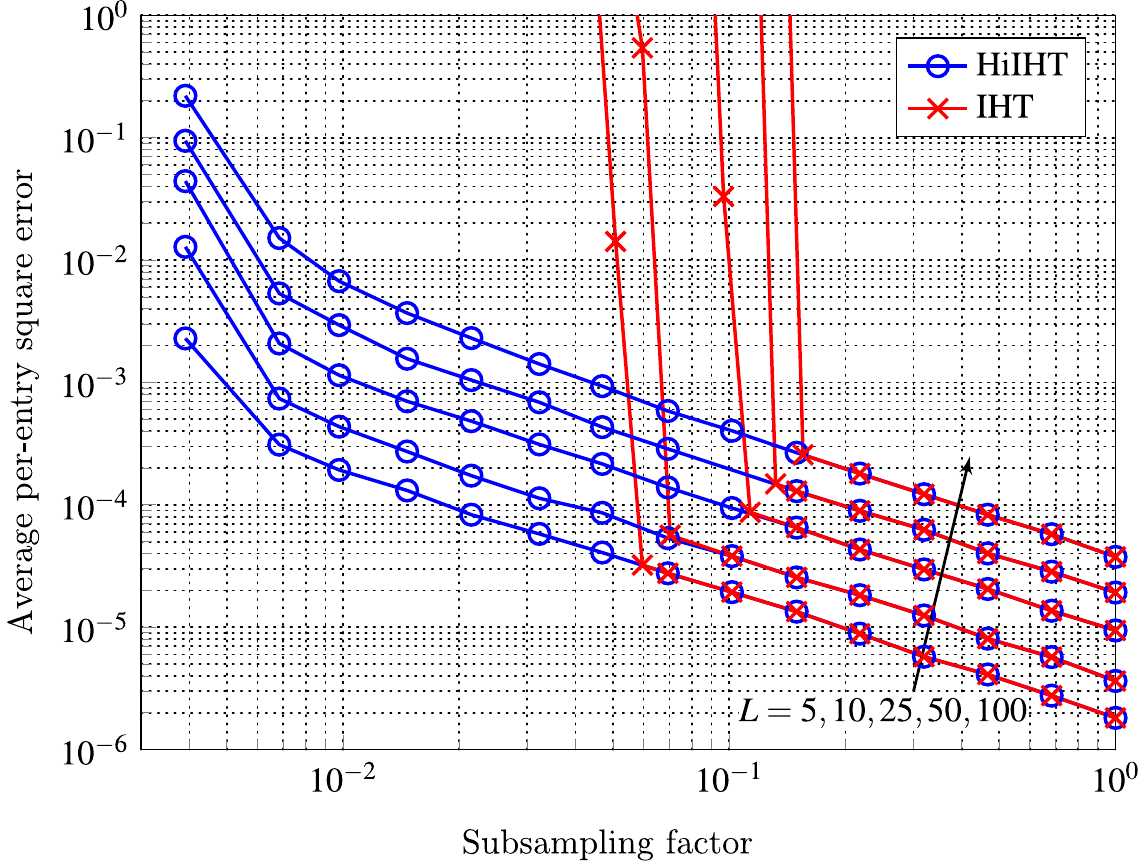}
    \caption{Comparision of HiIHT and IHT performance for channel matrix reconstruction. \textcopyright 2019 IEEE. Reprinted, with permission, from Ref.~\protect\citeWE{wunder2019low}.}
    \label{WEfig:MiMoHiIHT}
\end{figure}

In Fig.\ \ref{WEfig:MiMoHiIHT}, the performance of HiIHT and IHT 
are compared for $m=n=256$ and $N= 1024$. We generate data synthetically, and inject the measurement with Gaussian noise of an SNR $10$dB. The recovery quality is measured in terms of the mean per-entry square error $\frac{1}{nN}\Vert H-\hat{H} \Vert^2$ between the actual channel matrix $H$ and the estimate $\hat{H}$. This error is plotted against the sub-sampling factor $M/N$ for different values of $L$. We see that HiIHT handles a small sub-sampling factor considerably better than IHT. Indeed, only accessing one percent of the available antennas is enough to achieve reasonable performance with HiIHT, whereas IHT fails when less than about 10 percent of the antennas are utilized.

The communication setting presented here can be extended in several directions: First, we may drop the assumptions on the delays and angles to be on a grid  -- in the  off-the-grid case the vector $X$ is arguably still approximately sparse.  Second, we can model the case of multiple users by adding a third level to the hierarchical signal.  On this level sparsity naturally emerges assuming a sporadic user activity. 
We refer to Ref.\ \citeWE{wunder2019low} for details.

\subsection{Secure Massive Access}  \label{sec:blinddeconv}

With the rise of new communication technologies such as the internet of things (IoT) and tactile internet (TI), the amount of devices virtually explodes, and with it the amount of sensitive information gathered from various sensors and transmitted over the air. This development poses significant challenges on the security of communication channels and demands for new physical layers of security. In particular, it calls for fast and scalable low-overhead security schemes suitable for the frequent burst of spontaneous communication between low-complexity devices with a base station.
Here, we use the hierarchical measurement framework to design a secure massive access procedure based on blind deconvolution, see also the discussion on bi-sparse structures in Sec.~\ref{section:signal_model}. More details can be found in Ref.\  \citeWE{wunder2018secure}.

A base station sends out known pilots to enable all \emph{user equipments (UEs)} to measure the channel between the station and the UE. 
The channel is here modeled as a filter in $\K^{N}$, where $N$ is the length of the delay period. For each transmitting UE $p \in [N_d]$ and receiving base station antenna $q \in [N_r]$, there is one filter 
\begin{equation*}
h_{p,q}=(h_{p,q,1},\ldots, h_{p,q,i}, \ldots, h_{p,q,N})\in\K^{N}.
\end{equation*}
 The concrete appearance of the filters are again determined by delays caused by reflections on random physical features in the environment. Therefore, it is reasonable to assume that each $h_{p,q}$ is sparse, and, for fixed UE $q$, all channels $h_{p,q}$ for $p=1, \dots, N_t$ share the same sparsity pattern.

As in the previous section, the UE transmit their sequences $c_p \in \K^E$ by first linearly encoding them into signals $x_p=B_pc_p$ using a codebook $B_p \in \K^{N\times E}$, and then sending them over the channel. In an IoT scenario, the messages typically are very short, so that it can be assumed that they can be encrypted as sparse sequences $c_p$. During transmissions, these are convolved with the channel vectors, so that each of the base station's antennas receives a superposition of the UEs' signals, 
\begin{equation*}
\label{eq:conv_multi_user}
y_q = \sum\limits_{p=1}^{N_r} h_{p,q}\circledast (B_p c_p) +z_q 
\end{equation*} 
with $q=1,\ldots,N_t$ and $\circledast$ denoting the circular convolution. We may now \emph{lift} \citeWE{ling2017blind} the bilinear operation $(c_p,h_{p,q}) \to h_{p,q} \circledast B_pc_p$ to a linear operation $\mathrm{conv}_p: \K^{E\times N} \to \K^N$ on the matrix $b_ph_{p,q}^\top \in \K^{E \times N}$ resulting in 
\begin{equation}
\label{eq:lifted_bd_multi_user}
y_q = \sum\limits_{p=1}^{N_r} \mathrm{conv}_p(b_p h_{p,q}^\top) + z_q\, .
\end{equation}
We observe that the channel estimation task at the base station becomes 
the problem of simultaneously performing a blind deconvolution and de-mixing, naturally formalized as the linear reconstruction of a signal 
\begin{align*}
  X_q=(b_1 h_{1,q}^\top, \dots, b_{N_d} h_{N_d,q}^\top) \in (\K^{E \times N})^{N_d} \sim \K^{N_d\cdot E \cdot N}.
\end{align*}

The signal further exhibits the following structure:
Our assumptions of $\sigma$-sparse channels and $s$-sparse messages imply that the matrices $b_p h_{p,q}^\top$ are all $(s,\sigma)$-bisparse. 
As disscussed in Sec.~\ref{section:signal_model}, we may relax this to simple hierarchical $(s,\sigma)$-sparsity. 
Additionally assuming a sparse user activity at a given time, i.e. $b_p \neq 0$ only for $\mu$ users, the vector $X_q$ is a three-level $(s, \sigma, \mu)$-sparse vector.  
Note that the operator $\mathrm{conv}_p$ has a structure that is not covered by our theoretical results. Still, we may try to recover it using the HiHTP-algorithm.

We conduct simulations with
$N_t=1$ receive antenna and $N_r=10$ total UEs. We set $N=1024$ and $N=E=128$.
For each of the $N_r$ users a $\sigma$-sparse channel $h_k\in \R^{E}$ is drawn with the locations of the non-zeros distributed uniformly and entries drawn from the standard normal distribution. The signals are computed as $x_k = B c_k$ were $B\in \R^{N\times E}$ is a Gaussian random matrix and $c_k\in \R^E$ is $s$-sparse with values in $\{-1,1\}$ if the user is active, and $0$ if the user is not active. 
This results in the data $y_1\in\R^{N}$ as defined in \eqref{eq:lifted_bd_multi_user}. 

We vary the number of active users $\mu$, as well as the sparsities $s$ and $\sigma$. 
The Figs.~\ref{WEfig:success2}--\ref{WEfig:success5} below show the rate of successful recovery for varying number of active users, averaged over 20 runs per setup. The x- and\ y-axis show the channel sparsity $\mu$ and\ the signal sparsity $s$, respectively. As can be seen, the HiHTP-algorithm is indeed capable of recovering the ground truth, as long as the sparsity levels are low enough.
\begin{figure}[tb]
	\centering
    \begin{minipage}[b]{0.475\textwidth}
    \includegraphics[width=\textwidth]{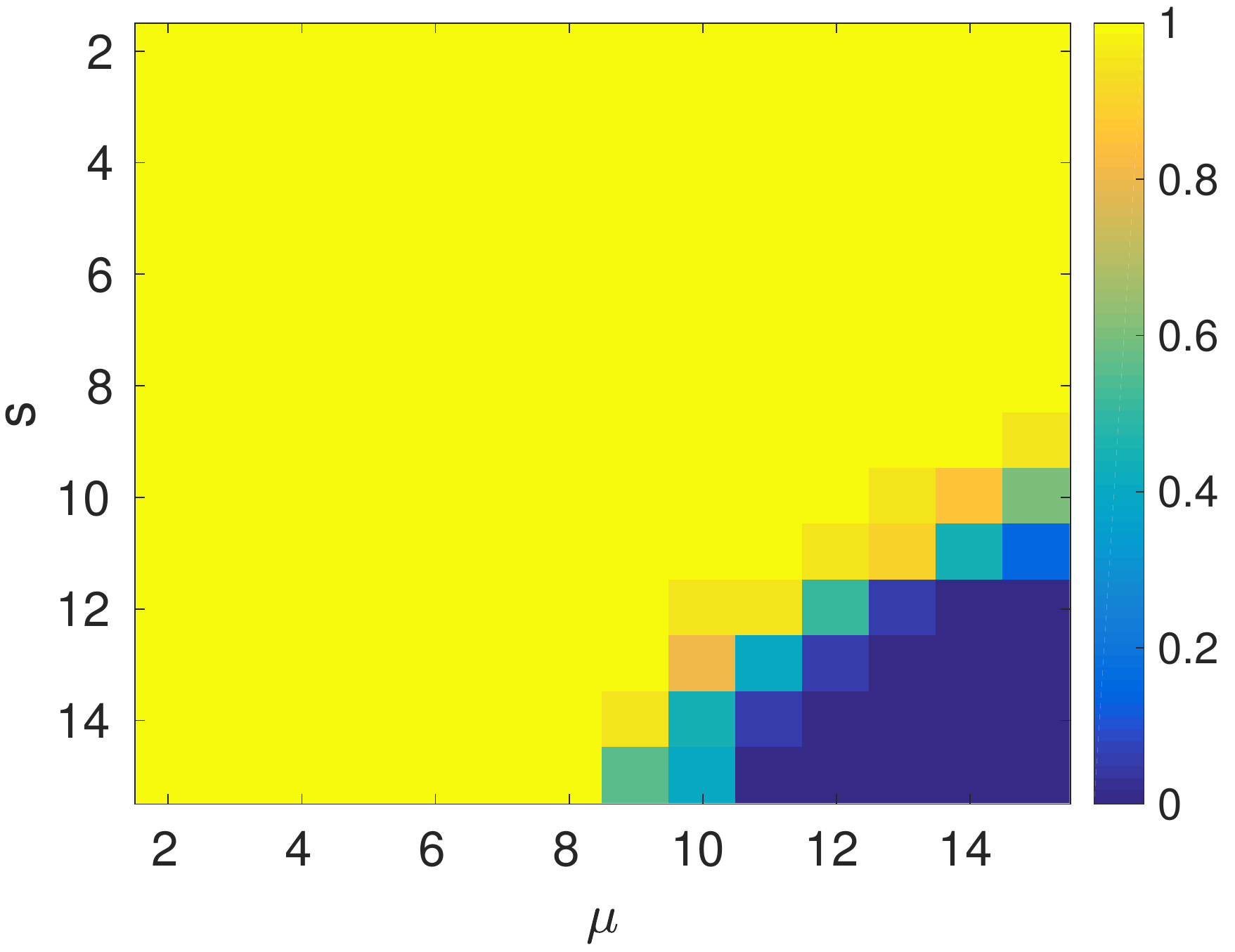}
    \caption{Recovery rate for 2 of 10 active users. 
    \textcopyright 2018 IEEE. Reprinted, with permission, from Ref.~\protect\citeWE{wunder2018secure}.}
      \label{WEfig:success2}
    \end{minipage}
    \hfill
    \begin{minipage}[b]{0.475\textwidth}
    \includegraphics[width=\textwidth]{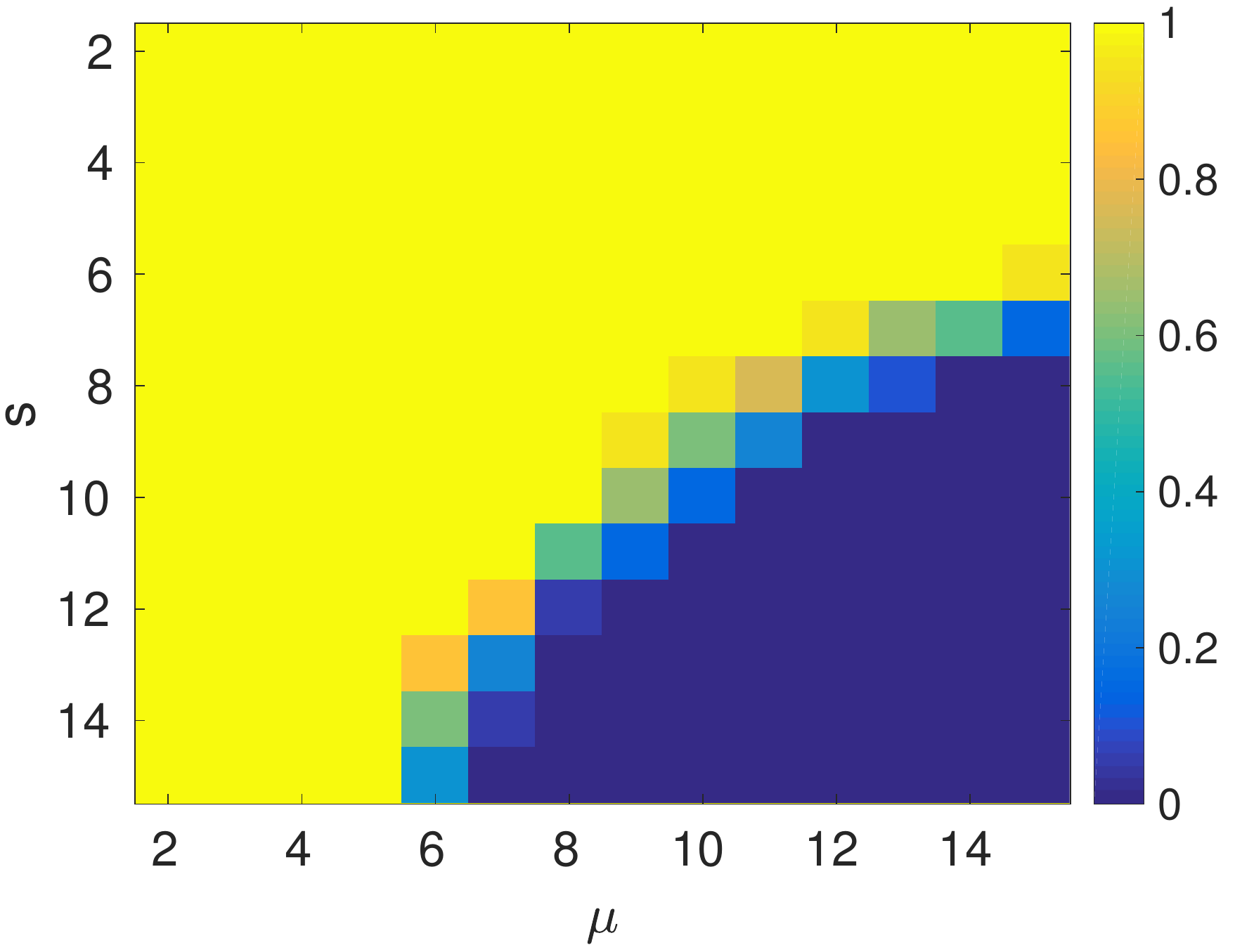}
    \caption{Recovery rate for 5 of 10 active users. 
    \textcopyright 2018 IEEE. Reprinted, with permission, from Ref.~\protect\citeWE{wunder2018secure}.}
    \label{WEfig:success5}
    \end{minipage}
\end{figure}

An interesting feature of the model is that it can be used to generate a secure communication scheme. 
To this end, we {make use of} the \emph{reciprocity} {of the channel}: The channel $h_{p,q}^{\uparrow}$ for transmission from UE $q$ to base station antenna $p$ is equal to the channel $h_{p,q}^{\downarrow}$ for transmission in the other direction. This reciprocity condition is fulfilled for modern off-the shelf WiFi devices \citeWE{reciprocity}. 
{Due to the reciprocity, the channel itself can serve} as a  source of shared randomness for the secret key generation. 
The communication protocol consists of two phases:

\noindent{\bf Phase 1}:
\begin{enumerate}
\item The base station sends a predefined pilot signal to all UEs.
\item Each UE $q$ measures the complex-valued channel gains $h_q^{\downarrow}= (h_{p,q}^{\downarrow})_{p \in [N_r]}$.
\item Each UE encrypts his/her message $m$ to a sequence $c_p=f(m,h_{q}^{\downarrow})$, using some encryption scheme $f$ and $h_{q}^{\downarrow}$ as a a random encryption key.
\end{enumerate}

\noindent{\bf Phase 2}
\begin{enumerate}
\item All the UEs $q$ send their encrypted \emph{sequences} $c_q$ to the base station using the scheme discussed above. The encoding operators $B_p$ are left public.
\item The base station receives the superposition of all the convolutions of the cipher text with the respective channels. {With a hierarchical thresholding algorithm, the station  inverts \eqref{eq:lifted_bd_multi_user}}, and, thus, gains knowledge of the cipher-texts $c_p$ and channels $h^{\uparrow}_{p,q}$. 
\item  Due to reciprocity $h^{\uparrow}_{p,q}= h^{\downarrow}_{p,q}$, the base station thereby obtains the encryption keys $h_q^{\downarrow}$, and decrypts the cipher-texts.
\end{enumerate}
The security of the scheme relies on the assumption that the channels of different users are independent of each other and can not be inferred from another position. Unless a man-in-the middle has access to the antenna of an UE, the eavesdropper can not use his/her channel coefficients to recover the message of another user.

We note that small variations between both channels, i.e.\ small violations of reciprocity, can be tolerated by adjusting the key generation process. One can for example quantize the channel gain sufficiently coarse to equalize the keys.
Here, the hierarchical framework is applied to solve a blind deconvolution and demixing problem. Refs.~\citeWE{gross2021hierarchical, wunder2021measure} present further examples of the hierarchical measurement framework applied to massive random access without a built-in security scheme.

\subsection{Blind quantum state tomography}
\label{sec:blind_quantum_state_tomography}

Quantum communication allows for the transmission of data under unprecedented levels of security \citeWE{RevModPhys.74.145}. Here, the security proofs are neither based on assumptions on the computational hardness of certain mathematical problems, nor on the feasibility of practically reverting or predicting the randomness of physical processes: Instead, there are proofs of security available based on the fundamental laws of nature themselves. Under mild assumptions, quantum key distribution can be proven secure under the 
most general attacks allowed by physics, within a
paradigm of closed laboratories. 
Simultaneously, the advent of novel quantum computing devices 
promises solving certain tasks with a significantly improved computational complexity 
compared to classical computing devices. 
These tasks include NP problems at the heart of established and universally employed cryptographic schemes. 
It is beyond the scope of the present article to introduce the various applications of the quantum technologies
\citeWE{Roadmap}. Instead, we here focus on a particular context in which hierarchical compressed sensing naturally comes into play: This is the task of semi-device dependently  identifying the 
state of a quantum device.
Methods for such characterization and certification tasks 
are important diagnostic tools in the development of quantum technologies. 
We refer to Refs.\ 
\citeWE{BenchmarkingReview, KlieschRoth:2020:Theory} for details.

The problem at hand here is the identification of quantum states prepared in some physical prescription. The recovery of unknown quantum states is called \emph{quantum state tomography}. 
A general quantum state is described by a trace-normalized, positive-definite complex matrix. Of particular interest are unit rank, so-called \emph{pure} quantum states or more generally 
low-rank quantum states.  
Ideally devices in quantum technologies operate or are envisioned to operate in pure quantum states of large dimensions. 
Quantum states of higher rank encode `classical' statistical mixtures of pure states typically produced by noisy operations. 
We denote the set of rank-$r$ quantum states by $\mathcal D^n_r \subset \C^{n\times n}$. 

An important diagnostic task for quantum devices is, thus, to learn the low-rank quantum state of the device from linear measurements. 
Exploiting the rank constraint on the quantum states in the recovery task is crucial to devise quantum tomography protocols working in state spaces of sizeable dimension.  
This renders compressed sensing method of crucial importance for quantum tomography \citeWE{Compressed,FlammiaGrossLiu:2012,ShabaniEtAl:2011,KalevKosutDeutsch:2015, SteffensEtAl:2017,RiofrioEtAl:2017,RothEtAl:2018:Recovering,QuantumReadout}.

That said, the apparata with which one performs the measurements can especially for near-term devices not be reasonably assumed to be 
fully characterized: Commonly there are calibrating parameters that are not fully known. An important practical problem is, thus, the recovery of a low-rank quantum state $\rho$ by means of measurement devices that are simultaneously themselves characterized by a
handful of parameters, giving rise to sparse
vectors $\xi$. 

In a linear approximation of the measurement device calibration, this leads to the problem of  \emph{blind (self-calibrating) quantum state tomography}:
Let $\mathcal A: \C^{nd\times d} \to \R^m$ be a linear map describing the measurement and calibration model.  
Given data $y = \mathcal A(X) \in \R^m$ and the linear map $\mathcal A$, recover $X$ under the assumption that 
\begin{equation}
	X \in \{\xi \otimes \rho \mid \text{$\xi \in \K^N$ $s$-sparse},\ \rho \in \mathcal D^n_r \} \subset \C^{Nn \times n}\, .
\end{equation}

The blind quantum state tomography problem can be regarded as a non-commutative analogon of 
 bisparse recovery problems where the data is bi-linear in two sparse vectors both to be recovered. 
Similarly to the vector case, already the projection onto the set of structured signal is an 
NP-hard problem. 
In fact, one can encode the \emph{sparse PCA problem} \citeWE{magdon-ismail_np-hardness_2017} and thereby CLIQUE into the task of finding the closest element of the form $\xi \otimes \rho$ with $\xi \in \K^N$, $\rho \in \mathcal D^n_r$ to a given $X \in \K^{Nn \times n}$ in Frobenius norm Ref.\ \citeWE[Theorem~ 3]{RothEtAl:2020:Semidevicedependent}.  
For this reason, it is not possible to directly derive an efficient algorithm based on a hard-thresholding operation for the blind quantum tomography problem.

\begin{figure*}[tb]
    \centering
\includegraphics[width=1\textwidth]{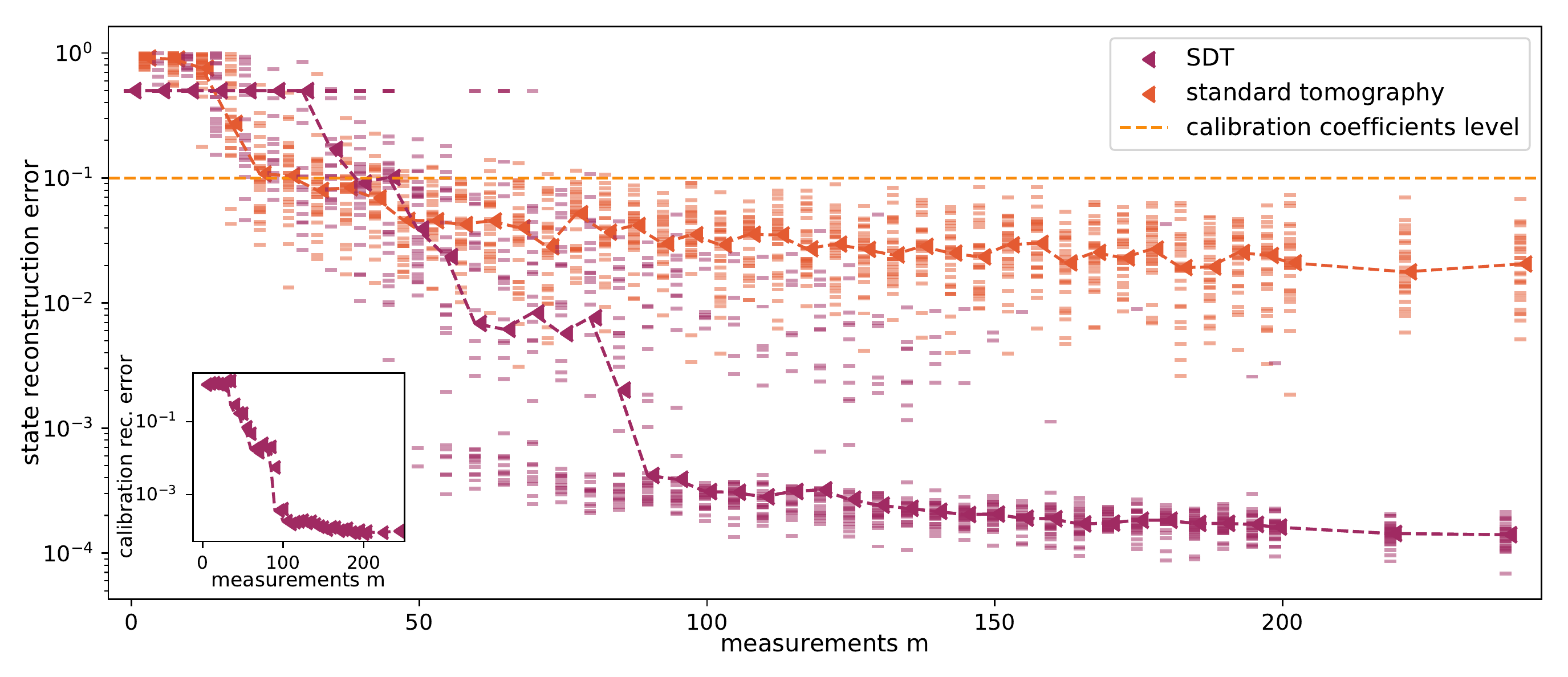}
	\vspace{-.7cm}
	\caption{ 
	\label{WEfig:numericsrandompauli}
	The figure (taken from Ref.~\protect\citeWE{RothEtAl:2020:Semidevicedependent}) displays the trace-norm reconstruction error for the SDT algorithm compared to a standard low-rank tomography algorithm for different number of measurements $m$ of sub-sampled random Pauli measurements. Each point depicts $30$ random measurement and signal instances with $r=1$, $d=8$, $n=10$ and $s=3$. The dotted lines indicate the median. The inline figure shows the mean $\ell_2$-norm reconstruction error of the calibration coefficients for the SDT algorithm.
	}
\end{figure*}

However, the problem of demixing a sparse sum of low-rank matrices introduced in Sec.~\ref{sec:sparseLowRank} can be seen as a relaxation to the closest hierarchically structured signal class that still allows for an efficient projection. The analogy to the relation of bi-sparsity and hierarchical sparsity is imminent.

Consequently, the SDT algorithm is a natural candidate to efficiently tackle the blind tomography problem. 
Fig.~\ref{WEfig:numericsrandompauli} shows numerical simulations of the performance of the SDT algorithms in the blind quantum tomography task for a random calibration model motivated by quantum technologies in comparison to a standard low-rank tomography algorithm. 
The relaxation to the hierarchical structured problem, however, comes at the cost of a sub-optimal scaling in complexity theory. 
While a parameter counting of the original blind tomography problem hints at an optimal scaling of $O(\max\{s \log N, n r\})$, the sparse demixing problem introduces already in parameter count an additional factor of $s$ to the second term $O(\max\{s \log N, s n r\})$. 
Due to the sparsity assumption on the calibration parameters, the total number of calibration parameter $N$ still only enters logarithmically.  
For this reason, the scheme remains highly scalable in practically relevant parameter regimes despite the relaxiation. 
At the same time, using the framework of hierarchical compressed sensing outlined above provides a rich toolkit to equip the SDT with flexible guarantees for many ensembles of measurement and calibration models. 
Another algorithmic approach to bi-linear structured problems such as the blind tomography problem is constraint alternating minimization. 
We refer to Ref.~\citeWE{RothEtAl:2020:Semidevicedependent} for further details. 

\section{Conclusion and outlook} \label{sec:outlook} 

In this chapter, we have introduced a framework for hierarchically compressed sensing with a focus mostly on the reconstruction of hierarchically sparse signals. In its core, standard approaches of compressed sensing naturally 
generalize to hierarchically structured signals, giving rise to recovery algorithms equipped with theoretical guarantees. 
Thereby, the successful recovery of hierarchically sparse signals via hard-thresholding algorithms can be established under a custom-tailored restricted isometry assumption. 
There are, however, 
a number of specific features that separates the hierarchical framework from its more generic counterparts.

At the heart of the approach is the fact that the projection operator onto the set of hierarchically structured signals is efficiently calculable via hierarchical hard-thresholding. 
Unlike for, e.g., the bisparse structure, it can be computed in linear time, and is highly amenable to parallelization. 
This in turn renders the simple recovery algorithm interesting in realistic parameter regimes and under practical demands.

Furthermore, within the hierarchical framework, there is a large family of operators that obey the hierarchical, but not the standard restricted isometry property. This makes the framework potentially applicable in settings where standard compressed sensing is infeasible.

On a more theoretical level, the hierarchically sparse structure can be used as a relaxation of the complicated bi-sparse structure. In particular, we have presented numerical evidence that instances of the sparse blind deconvolution problem can be solved using HiHTP. 
And we have invoked the same strategy for the quantum tomography problem and other related questions. 
While in this context theoretical guarantees are expected to be sub-optimal, the simplicity and flexibility of the hierarchical framework might still be of merit in order to analyze complicated measurement settings.  
We leave further exploring these matters to future research. 
A particularly interesting question is to analyze the HiRIP properties of the blind deconvolution operator. 

Indeed, we have at the end of this chapter seen several exemplary applications where the hierarchical approach facilitates recovery. This brings us to the arguably most important feature of the framework: Hierarchically structured signals naturally emerge in many applications. 
From our own background and past research, we can conclude this with some confidence. 
But of course, we very much suspect that there are many applications we are unaware of where the hierarchical framework is readily applicable. 
For the sake of clarity, we have mainly focused our exposition on the set of two-level hierarchically sparse vectors and merely hinted at the generalizations towards multiple levels potentially mixing low-rankness, sparsity and potentially even further structures that for themselves come with an efficient projection. 
We hope that we have conveyed that the approach, and even most of the results we presented, can be rather straight-forwardly generalized to this rich family of hierarchical signal structures, leaving the playing field wide open.

 \subsection*{Acknowledgements}
 This work is a report of some of the  findings of the DFG-funded project
 EI 519/9-1 within  the priority 
 program `Compressed Sensing in Information Processing' (CoSIP), 
jointly held by J.\ Eisert and G.\ Wunder. We specifically thank our coauthors, in particular
M. Barzegar,
G. Caire,
R. Fritscheck,
S. Haghighatshoar,
D. Hangleiter,
M. Kliesch, 
S. Stefanatos,
R. Kueng, and
J. Wilkens,
with which we have explored this research theme over the years.

\bibliographystyleWE{spmpsci}                 %% 

\bibliographyWE{bibfileWE}                %

\end{document}